\pdfoutput=1
\documentclass[10pt, conference, compsocconf]{IEEEtran}

\newcommand{\frage}[1]{}

\usepackage{algorithm2e}
\usepackage{amsthm}
\usepackage[english]{babel}
\usepackage{caption}
\usepackage{cite}
\usepackage{enumerate}
\usepackage{etoolbox}
\usepackage{graphicx}
\usepackage{lipsum}
\usepackage{xcolor}
\usepackage{subfigure}

\providetoggle{long}
\settoggle{long}{true}

\newcommand{\Ohsmall}[1]{\mathcal{O}(#1)}
\newcommand{\Oh}[1]{\mathcal{O}\left(#1\right)}

\newtheorem{theorem}{Theorem}
\newtheorem{lemma}{Lemma}

\begin{document}
	\pagestyle{plain}
	\theoremstyle{definition}
	\newtheorem*{inv*}{Invariants}	
	
	\title{Fast Parallel Operations on Search Trees}
	
	\author{\IEEEauthorblockN{Yaroslav Akhremtsev}
		\IEEEauthorblockA{Institute of Theoretical Informatics\\
			KIT\\
			Karlsruhe, Germany\\
			Email: yaroslav.akhremtsev@kit.edu}
		\and
		\IEEEauthorblockN{Peter Sanders}
		\IEEEauthorblockA{Institute of Theoretical Informatics\\
			KIT\\
			Karlsruhe, Germany\\
			Email: peter.sanders@kit.edu}
	}
	
	\maketitle
	\begin{abstract}
		Using $(a,b)$-trees as an example, we show how to perform
		a parallel split with logarithmic latency
		and parallel join, bulk
		updates, intersection, union (or merge), and (symmetric) set difference
		with logarithmic latency and with information theoretically optimal
		work. We present both asymptotically optimal solutions and simplified
		versions that perform well in practice -- they are several times faster
		than previous implementations.
	\end{abstract}
	\section{Introduction}
	
	Sorted sequences that support updates and search in logarithmic time
	are among the most versatile and widely used data structures. For the most frequent case
	of elements that can only be compared, search trees are the most
	widely used representation. When practical performance is an issue,
	$(a,b)$-trees are very successful since they exhibit
	better cache efficiency than most alternatives.
	
	Since in recent years Moore's law only gives further improvements of
	CPU performance by allowing machines with more and more cores, it has
	become a major issue to also parallelize data structures such as
	search trees. There is abundant work on \emph{concurrent data
		structures} that allows asynchronous access by multiple threads
	\cite{LehYao81,HLLS06, BCCOK10}. However, these approaches are not scalable in
	the worst case, when all threads try to update the same part of the
	data structure. Even for benign inputs, the overhead for locking or
	lock-free thread coordination makes asynchronous concurrent data
	structures much slower than using \emph{bulk-operations}
	\cite{FriasS07,sjc-palmbmcp-11,EKS14} or operations manipulating
	entire trees \cite{DBLP:journals/corr/BlellochFS16}.
	We concentrate on bulk operations here and show in Section~\ref{furtherOps}
	that they can be reduced to tree operations.\iftoggle{long}{%
		\footnote{It is less clear how to go the opposite way -- viewing bulk
			operations as whole-tree operations -- in the general mixed case
			including interactions between operations.}%
	}{}
	The idea behind bulk
	operations is to perform a batch of operations in parallel.  A
	particularly practical approach is to sort the updates by key and to
	simultaneously split both the update sequence and the search tree in
	such a way that the problem is decomposed into one sequential bulk
	update problem for each processor \cite{FriasS07,EKS14}. Our main
	contribution is to improve this approach in two ways making it
	essentially optimal. Let $m$ denote the size of the sequence to be
	updated and $k\leq m$ denote the number of updates.  Also assume that the
	update sequence is already sorted by key.  On the one hand, we reduce
	the span for a bulk update from $\Ohsmall{\log k\log m}$ to $\Ohsmall{\log m}$.
	On the other hand, we reduce the work from $\Ohsmall{k\log m}$ to
	$\Ohsmall{k\log\frac{m}{k}}$ (to simplify special case treatments for the case $k\approx m$,
	in this paper we define the logarithm to be at least one) which
	is information-theoretically optimal
	in the comparison based model. After introducing the
	sequential tools in Section~\ref{sec:preliminaries}, we present logarithmic time
	parallel algorithms for splitting and joining multiple $(a,b)$-trees in
	Sections~\ref{sec:parallel_split} and \ref{sec:parallel_join}
	respectively.  These are then used for a parallel bulk update with the
	claimed bounds in Section~\ref{sec:bulk_updates}.
	For the detailed proofs we refer to the full version of the paper.
	
	\subsection*{Related Work}
	We begin with the work on sequential data structures.
	Kaplan and Tarjan described finger trees in~\cite{Kaplan96purelyfunctional} with
	access, insert, and delete operations in logarithmic time, and joining of two trees in
	$\Ohsmall{\log\log m}$ time. Brodal et al. described a catenable sorted lists \cite{Brodal2006purelyfunctional}
	with access, insert, and delete operations in logarithmic time, and combination of two lists in
	worst case constant time. The authors state that it is hard to implement a split operation and
	it will lead to access, insert and delete operations in $\Ohsmall{\log m \log\log m}$ time.
	
	A significant amount of the research has been done for operations on pairs of trees, in particular,
	union, intersection and difference. A lower bound for the union operation
	is $\Omega(k\log\frac{m}{k})$ in the comparison based model~\cite{Brown:1979:FMA:322123.322127}.
	Brown and Tarjan~\cite{Brown:1979:FMA:322123.322127} presented an optimal union algorithm for AVL trees and $2$-$3$ trees.
	They also published an optimal algorithm
	for union level-linked $2$-$3$ trees~\cite{DBLP:journals/siamcomp/BrownT80}.
	The same results were achieved for the level-linked $(a, b)$ trees~\cite{DBLP:journals/acta/HuddlestonM82}.
	
	Paul, Vishkin and Wagener gave the first parallel algorithm for 
	search, insertion and deletion algorithms for $2$-$3$ trees on EREW PRAMs~\cite{paul1983parallel}.
	The same result was achieved for B-trees~\cite{Higham1994329}
	and for red-black trees~\cite{Park2001415}. All these algorithms perform $\Ohsmall{k\log m}$ work.
	The first EREW PRAM union algorithm with $\Ohsmall{k\log\frac{m}{k}}$ work 
	and $\Ohsmall{\log m}$ span
	was given
	by Katajainen et al.~\cite{9c429ad074cd11dbbee902004c4f4f50}. But this algorithm contains a false proposition
	and the above bounds do not hold~\cite{Blelloch:1998:FSO:277651.277660}.
	Blelloch et al.~\cite{Blelloch:1998:FSO:277651.277660} presented a parallel union
	algorithm with expected $\Ohsmall{k\log \frac{m}{k}}$
	work and $\Ohsmall{\log k}$ span for the EREW PRAM with scan operation. This implies
	$\Ohsmall{\log^2 m}$ span on a plain EREW PRAM.
	Recently, they presented a framework that implements union set operation 
	for four balancing schemes of search trees~\cite{DBLP:journals/corr/BlellochFS16}.
	Each scheme has its own join operation; all other operations are implemented using it.
	Experiments in \cite{DBLP:journals/corr/BlellochFS16} indicate that our
	algorithms are faster -- probably because their implementations are
	based on binary search trees which are less cache efficient than
	$(a,b)$-trees.

	\section{Preliminaries}\label{sec:preliminaries}
	
	We consider weak $(a, b)$-trees, where $b \ge 2a$ \cite{DBLP:journals/acta/HuddlestonM82}.
	A search tree $T$ is an $(a, b)$-tree if%
	\iftoggle{long}{%
		\begin{itemize}
			\item all leaves of $T$ have the same depth
			\item all nodes of $T$ have degree not greater than $b$
			\item all nodes of $T$ except the root have degree $\geq a$
			\item the root of $T$ has degree not less than $\min(2, |T|)$
			\item the values are stored in the leaves
		\end{itemize}
	}{%
	~(1) all leaves of $T$ have the same depth; (2) all nodes of $T$ have degree not greater than $b$;
	(3) all nodes of $T$ except the root have degree $\geq a$; (4) the root of $T$ has degree not less
	than $\min(2, |T|)$; (5) the values are stored in the leaves;
}
Let $m$ be an upper bound of the size of all involved trees.
We denote the parent of some node $n$
by $p(n)$. The \textit{rank} $r(n)$ of the node $n$ is the number of nodes (including $n$) on the path
from $n$ to any leaf in its subtrees. 
We define the rank of a tree $T$, denoted $r(T)$, to be the rank of its root.
We denote the left-most (right-most) path from the root to the
leaf as the left (right) \emph{spine}. We employ two operations to process the nodes of the tree.
The \textit{fuse} operation fuses nodes $n_1$, $n_2$ using a \textit{splitter key} --
this key $\le$ any key in $n_1$ and $\ge$ any key in
$n_2$ -- into a node $n$. The \textit{split} operation splits a node $n$ into two nodes $n_1$, $n_2$ and
a splitter key such that $n_1$ contains the first $\lfloor\frac{d}{2}\rfloor$ (here $d$ is the 
degree of $n$) children of $n$,
$n_2$ contains remaining $\lceil\frac{d}{2}\rceil$ children of $n$, and the $\lfloor\frac{d}{2}\rfloor$th
key of $n$ is the splitter key.

Here we explain the basic algorithms for joining two trees and splitting a tree
into two trees that are basis of all our algorithms.
A more detailed description can be found in \cite{DBLP:books/sp/MehlhornS2008}. 
\paragraph{Joining Two Trees}\label{paragraph:basic_join}
We now present an algorithm to join two $(a, b)$-trees
$T_1$ and $T_2$ such that all elements of $T_1$ are less than
or equal to
the elements of $T_2$ and $r(T_1) \ge r(T_2)$. This
algorithm joins the trees $T_1$ and $T_2$
into a tree $T$ in time
\iftoggle{long}{%
	$\mathcal{O}(r(T_1) - r(T_2) + 1) = \mathcal{O}(\log (\max(|T_1|,
	|T_2|))) = \mathcal{O}(\log |T|)$.
}{%
$\mathcal{O}(r(T_1) - r(T_2) + 1) = \mathcal{O}(\log (\max(|T_1|,
|T_2|)))$.
}
Our main goal is to ensure
that the resulting tree is balanced. First, we descend $r(T_1) - r(T_2)$
nodes on the right spine until we reach the node $n$ such that $r(n) = r(T_2)$.
Next, we choose the largest key
in $T_1$ as the splitter. If the degree of the
root of $T_2$ or $n$ is less than $a$ then we fuse them into $n$.
If the degree of $n$ after the fuse is $\le$ $b$ then the join operation ends.
Otherwise, we split $n$ into
$n$ and the root of $T_2$ and update the splitter key (if
necessary). The degrees of $n$ and the root of $T_2$ are less than $b$ after the split, since we fuse
them if at least one of them has degree less than $a$. We insert the splitter key and the pointer
to the root of $T_2$ into $p(n)$. Further, the join operation proceeds as 
an insert operation \cite{DBLP:books/sp/MehlhornS2008}. We propagate splits up the right
spine until all nodes have degree less than or equal to $b$
or a new root is created. The case $r(T_1) \le r(T_2)$ is handled similarly.

\paragraph{Sequential Split}\label{paragraph:basic_split}
We now describe how to split an $(a, b)$-tree $T$ at a given element $x$ 
into two trees $T_1$ and $T_2$, such that all elements
in the tree $T_1$ are $\le$ $x$, and all elements in the tree $T_2$ are $\ge$ $x$.
First, we locate a leaf $y$ in $T$ containing an element of minimum value in $T$
greater than $x$. 
Now consider the path from the root to the leaf $y$. We split each node $v$ on the path
into the two nodes
$v_{\mathrm{left}}$ and $v_{\mathrm{right}}$, such that $v_{\mathrm{left}}$ contains all children of $v$
less than or equal
to $x$ and $v_{\mathrm{right}}$ the rest. We let $v_{\mathrm{left}}$ and $v_{\mathrm{right}}$ be the roots of
$(a, b)$-trees.
Next, we continue to join all the left trees with the roots $v_{\mathrm{left}}$ among the path from
the leaf $y$
to the root of $T$ using the join algorithm described above.
As the result we obtain $T_1$. The same join operations are performed with the
right trees, which give us $T_2$.
All join operations can be performed
in total time
$\mathcal{O}(\log |T|)$%
\iftoggle{long}{%
	, since the left and right trees have increasing height.
	Consider the roots $v_1, v_2, \dots, v_k$ of the left trees and their corresponding ranks
	$r_1, r_2, \dots, r_k$, where $r_1 \le r_2 \le \dots \le r_k$. We first join trees with roots
	$v_1$ and $v_2$.
	Next, we join $v_3$ with the result of the previous join, and so on. The first join operation
	takes time
	$\mathcal{O}(r_2 - r_1 + 1)$, the next one takes time $\mathcal{O}(r_3 - r_2 + 1)$.
	Thus, the total time is
	$\sum_{i = 1}^{k - 1} \mathcal{O}(r_{i + 1} - r_i + 1) = \mathcal{O}(k + r_k - r_1) =
	\mathcal{O}(\log |T|)$.
}{.}%

\paragraph{Sequential Union of a Sorted Sequence with an $(a, b)$-tree}
\label{paragraph:seq_merge}
Here we present an algorithm to union an $(a, b)$-tree
and a sorted sequence $I = \langle i_1, \dots, i_k\rangle$
This algorithm is similar to the algorithm
described in~\cite{Brown:1979:FMA:322123.322127}.
Let $l_j$ denote the leaf where element $i_j$ will be inserted
(i.e., the leaf containing the smallest element with key $\geq i_j$).
First, we locate $l_1$ by following the path from the root of $T$ to $l_1$
and saving this root-leaf path on a stack $P$.
When $l_1$ is located, we insert $i_1$ there
(possibly splitting that node and generating a splitter key to be inserted in the parent).
Next, we pop elements from $P$ until we have found the lowest common ancestor of $l_1$ and $l_2$.
We then reverse the search direction now searching for $l_2$. We repeat this process until 
all elements are inserted. 
We visit $\Ohsmall{k\log\frac{m}{k}}$ nodes and perform $\Ohsmall{k\log\frac{m}{k}}$ splits
during the course of the algorithm according to Theorems 3 and 4 from~\cite{DBLP:journals/acta/HuddlestonM82}.
Hence, the total work of the algorithm is $\Ohsmall{k\log\frac{m}{k}}$ even without using level-linked
$(a, b)$-trees as in~\cite{DBLP:journals/acta/HuddlestonM82}.

\section{Parallel Split}\label{sec:parallel_split}
The parallel split algorithm resembles the sequential version, but we need to split a
tree $T$ into $k$ subtrees
$T_1, \dots, T_{k}$ using a sorted sequence of separating keys $S = \langle s_1, \dots, s_{k}\rangle$,
where tree $T_i$ contains keys greater than $s_{i - 1}$ and less or equal than $s_i$ 
(we define $s_0 =-\infty$ and $s_k = \infty$ to avoid special cases).
For simplicity we assume
that $k$ is divisible by $p$ (the number of processors). If $k > p$ then we split the tree $T$ into $p$
subtrees according to the subset
of separators $S' = \langle s_{k/p}, s_{2k/p}, \dots, s_{(p - 1)k/p}\rangle$.
Afterwards, each PE (processor element) performs 
$k/p - 1$ additional sequential splits to obtain $k$ subtrees. From this point on we assume that $p = k$.

\begin{theorem} \label{th:parallel_split}
	We can split a tree $T$ into $k$ trees with $\Ohsmall{k \log |T|}$ work and $\Ohsmall{\log |T|}$ span.
\end{theorem}

Note that the algorithm is non-optimal -- it performs
$\Ohsmall{k \log |T|}$ work whereas the best sequential
algorithm splits a tree $T$ into $k$ subtrees in $\Ohsmall{k \log \frac{|T|}{k}}$ time.

We now describe the parallel split algorithm. First, the PE $i$ locates a leaf $l_i$ for each $s_i \in S$,
which contains the maximum element in $T$ less than or equal to $s_i$. Also, we save a first node
$r_i$ on the path from the root, where $l_{i - 1}$
and $l_i$ are in different subtrees. 
For $r_1$ this means a dummy node above the root.
Next, PE $i$ copies all nodes on the path from $l_i$ to $r_i$, but only keys $\leq  s_i$ and their
corresponding children. We consider these nodes to be the roots of $(a, b)$-trees and join
them as in the sequential
split algorithm. We can do this in $\mathcal{O}(\log m)$ time, as these trees have monotone or strictly
increasing ranks.
Let us refer to the resulting tree as $T_{\mathrm{right}}$. The same actions can be done on the path
from $l_{i - 1}$ to $r_i$,
except that we copy elements greater than $s_{i - 1}$ to new nodes. After joining
the new nodes we obtain 
a tree $T_{\mathrm{left}}$. We also build a tree $T_{\mathrm{central}}$ from the keys and corresponding children
of the node $r_i$ that are in the
range $(s_{i - 1}, s_i]$.  The last step is to join the trees $T_{\mathrm{left}}$,
$T_{\mathrm{central}}$ and $T_{\mathrm{right}}$. 

These operations can be done in parallel for each $i$,
since all write operations are performed on copies of the nodes owned by the processor performing
the respective operations. 
When we finish building the trees $T_1, \dots, T_k$ (we
use a barrier synchronization
to determine this in $\mathcal{O}(\log k)$ time) we erase all nodes
on the path in $T$ from the root to each
leaf $l_i$. Each PE $i$ then erases all nodes on the path from leaf $l_i$ to $r_i$, excluding $r_i$.

Each PE locates necessary leaf, builds trees $T_{\mathrm{left}}$, $T_{\mathrm{right}}$
by traversing up and down a path not longer than $\Ohsmall{\log |T|)}$ nodes.
Also each PE erases a sequence
of nodes not longer than the height of the tree $T_{\mathrm{right}}$, 
which is $\mathcal{O}(\log |T|)$. Therefore, Theorem~\ref{th:parallel_split} holds.

\section{Parallel Join} \label{sec:parallel_join}
We describe how to join $k$ trees $T_1, \dots, T_k$, where $m = \sum_i |T_i|$ and $p = k$, 
since this is the most interesting case.
When joining $k>p$ trees, we can assign $\leq \lceil k/p\rceil$ trees to each PE. 
First, we present a non-optimal parallel join algorithm.
Next, we present a modified sequential join algorithm.
Finally, we construct an optimal parallel join algorithm that combines
the non-optimal parallel join and the modified sequential join algorithms.

\begin{theorem} \label{th:optimal_parallel_join}
	We can join $k$ trees with $\Ohsmall{k \log \frac{m}{k}}$ work and $\Ohsmall{\log k + \log m}$ parallel time
	using $k$ processors on a CREW PRAM.
\end{theorem}

Note that the algorithm is optimal, since the best sequential algorithm joins $k$
$(a, b)$-trees in $\Ohsmall{k \log \frac{m}{k}} $ time~\cite{Moffat}. 

\subsection{Non-optimal parallel join.} \label{subsec:non_opt_join}
Let us first explain the basics of the parallel algorithm.
The simple solution is to join
pairs of trees in parallel (parallel pairwise join).
After each group of parallel join operations, the number of trees halves. Hence, each tree
takes part in at most $\lceil\log k\rceil$ 
join operations. Each join operation takes $\mathcal{O}(\log m)$ time.
That is, we can join $k$ trees in time $\Ohsmall{\log k \log m}$.

We improve this bound to $\mathcal{O}(\log m + \log k)$ by reducing
the time for a join operation to a constant by
solving the two following problems in constant time: 
finding a node with a specific rank on a spine of a tree and
performing a sequence of splits of nodes with degree $b$.

We describe the first issue in Section~\ref{subsubsec:node_retrieve_par}:
we can retrieve a spine node with specified rank using
an array $a$ where the $a[i]$ points to the spine node with rank $i$. 
The only challenge is to keep this array up to date in constant time.

We describe the second issue in Section~\ref{subsubsec:split_b_degree_seq}:
how to perform all required node splits in constant time.
The main observation is that if there are no nodes of degree $b$ on the
spines then there are no node splits during the course of the algorithm. 
So we preprocess
each tree $T_i$ such that there are no nodes of degree~$b$ on their left/right
spines by traversing each tree in a bottom up fashion.
The preprocessing can be done in parallel in $\Ohsmall{\log m}$ time.
Now consider the task of joining a sequence of preprocessed trees. The only
nodes of degree $b$
that could be on the left/right spines will appear during the course of the
join algorithm.
We can take advantage of this fact
by assigning a dedicated PE to each node of degree $b$.
This PE will split the node 
when needed.

We describe how to maintain the sizes of the subtrees in Section~\ref{subsubsec:subtree_sizes}.
This allows us
to search for the $i$-th smallest element in a tree $T$ in $\mathcal{O}(\log |T|)$ time.
We use the search of the $i$-th smallest element in Section~\ref{sec:selection}.

Note that we dedicate several tasks to one PE during the course of the algorithm but
this does not affect the resulting time bound.
Finally, we combine the above ideas
into the modified join algorithm in~\ref{subsubsec:join_algorithm}.
This algorithm joins $k$ trees in $\Ohsmall{\log k + \log m}$ time.

\begin{lemma}
	\label{lemma:nonoptimal_parallel_join}
	We can join $k$ trees in $\Ohsmall{\log k + \log m}$ time using $k$ processors on a CREW PRAM.
\end{lemma}

We explain how to obtain the result of Lemma~\ref{lemma:nonoptimal_parallel_join} 
in Section~\ref{subsubsec:join_algorithm}.
As a result we obtain an algorithm that performs join of $k$ trees
in $\mathcal{O}(\log m + \log k)$ time, has work $\mathcal{O}(k \log m)$,
and consumes $\Ohsmall{k \log m}$ memory.

\subsubsection{Fast Access to Spine Nodes by Rank}
\label{subsubsec:node_retrieve_par}

Suppose we need to retrieve a node with a certain rank on the right spine of a tree $U$.
The case for the
left spine is similar. We maintain an array of pointers to the nodes on the right spine of $U$
such that the $i$-th element of the array points
to the node with rank $i$ on the right spine.
\iftoggle{long}{%
	See Figure~\ref{fig:join_preliminaries_a}.
}{}%
We build this array during the preprocessing step.
We can retrieve a node by its rank in constant time with such an array.
The only problem is that after a join of $U$ with another tree some pointers
of the array point to nodes that are not on the right spine.
We describe how to maintain the pointers to the spine nodes throughout the join operations up-to-date.

Suppose that we have joined two trees $U$ and $V$, where $r(U) \ge r(V)$.
Let $R_U$ and $R_V$ denote the arrays of the pointers to the nodes on the right spines of $U$ and $V$ respectively.
The first $r(V)$ pointers of $R_U$ point to the nodes that are not on the
right spine anymore. Consequently, nodes that were on the right spine of the tree $V$ are on the right
spine of $U$ now. The first $r(V)$ elements of $R_V$ point to the nodes on the right spine of $U$ with 
the ranks in $[1, r(V)]$, and elements of $R_U$ with the indices in $(r(V), r(U)]$ 
point to the nodes on the right spine of $U$ with the rank in $(r(V), r(U)]$.
Hence, the interval $[1, r(U)]$ is split into the two subintervals: $[1, r(V)]$ and $(r(V), r(U)]$.
See Figures~\ref{fig:join_preliminaries_b} and \ref{fig:join_preliminaries_c}.

Now we show how to retrieve a node by its rank in constant time during a sequence of the join operations.
First, we explain how to maintain the arrays with up-to-date pointers to the right spine after the join operation.
Suppose we join trees $U$ and $V$ and we need to know the node $n \in U$ where $r(n) = r(V)$.
We maintain stacks $S_U$ and $S_V$ for trees $U$ and $V$ respectively, which are
implemented using linked lists.
Each element of these stacks is a pair $(R, I)$, where $R$ is an array 
with pointers to the right spine of the corresponding tree. $I$ is
an interval of the indices of the elements in $R$, such that $R[i]$ $(\forall i \in I)$
points to a node on the right spine of the tree.
We maintain the following invariants for each tree $T_i$ and its corresponding stack 
during the course of the algorithm:
\begin{enumerate}
	\item \label{inv_1}
	The stack $S_i$ contains disjoint intervals
	$[r_i, r_{i + 1} - 1]$, for $i = 1, \dots, |S_i|, r_1 = 1, r_{|S_i| + 1} = r(T_i) + 1$.
	These intervals are arranged in sorted order in $S_i$.
	
	\item \label{inv_2}
	$\bigcup_{(R, I) \in S_i} I = [1, r(T_i)]$.
	
	\item \label{inv_3}
	The element $R[i]$ $(i \in I)$ points to the node $n$ on the
	right spine where $r(n) = i$ for $\forall ~(R, I) \in S_i$.
\end{enumerate}

\begin{figure}[h]
	\centering
	\iftoggle{long}{%
		\subfigure[]{\includegraphics[scale = 0.75]{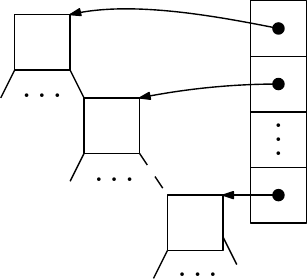}\label{fig:join_preliminaries_a}}
		\hfil
	}{}%
	\subfigure[]{\includegraphics[scale = 0.85]{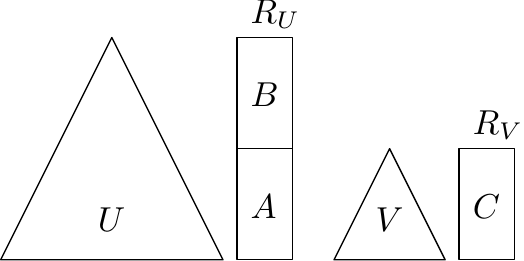}\label{fig:join_preliminaries_b}}
	\hfil
	\subfigure[]{\includegraphics[scale = 0.85]{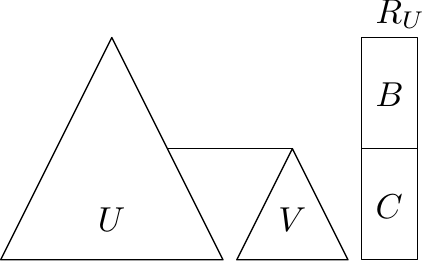}\label{fig:join_preliminaries_c}}
	\caption{Letter $A$ denotes the first $r(V)$ elements of $R_U$.
		Letter $B$ denotes the last $r(U) - r(V)$ elements of $R_U$. 
		Letter $C$ denotes the first $r(V)$ elements of $R_V$.}
	\label{fig:join_preliminaries}
\end{figure}
\iftoggle{long}{%
	Let us first consider the simple case, where each stack contains only a single element.
	Next, we extend it to the general case.
	The stack $S_U$ contains a pair $(R_U, [1, r(U)])$, and the stack $S_V$
	contains a pair $(R_V, [1, r(V)])$ after the preprocessing step. The invariants are true for $S_U$ and $S_V$.
	After the join operation we add the element of $S_V$ on the top of $S_U$.
	Consequently, the $S_U$ contains $(R_V,[1, r(V)])$ and
	$(R_U, (r(V), r(U)])$, and the invariants hold.
	
	Now we describe the general case. Suppose that the stacks $S_U$ and 
	$S_V$ contain more than one pair
	and they satisfy the invariants from above.
}{%
Suppose that the stacks $S_U$ and $S_V$ are not empty
and they satisfy the invariants.
}%
We need to find a node with rank $r(V)$ on the right spine of $U$.
First, we search for the pair $(R, I) \in S_U$, such that $r(V) \in I$. 
Let $S_U$ contains pairs with the following intervals:
$[r_i, r_{i + 1} - 1]$, where $i = 1, \dots, |S_U|, r_1 = 1, r_{|S_U| + 1} = r(U) + 1$.
We pop pairs from the stack $S_U$
until we find a pair with interval $I = [r_j, r_{j + 1} - 1]$
such that $r(V) \in I$. 
After the join operation we push the elements of $S_V$ on the top of $S_U$.
We refer to this operation as the \textit{combination} of $S_U$ and $S_V$. 
Consequently, stack $S_U$ contains all pairs of $S_V$ and pairs with intervals
$[r(V) + 1, r_{j + 1} - 1], [r_i, r_{i + 1} - 1]$, where $i = j + 1, \dots, |S_U|$.
We do not add the pair with the interval $[r(V) + 1, r_{j + 1} - 1]$ to the stack $S_U$ if
$r(V) + 1 > r_{j + 1} - 1$.

Now we show that the invariants hold for the resulting stack $S_U$.
The union of the pairs in $S_V$ is $[1, r(V)]$ and all intervals in the stack $S_V$ are disjoint and sorted.
The intervals $[r(V) + 1, r_{j + 1} - 1], [r_i, r_{i + 1} - 1]$, where $i = j + 1, \dots, |S_U|$,
are disjoint, sorted, and their union is $[r(V) + 1, r(U)]$.
Hence, Invariants~\ref{inv_1} and~\ref{inv_2} are hold.
We also have popped all pairs $(R, I) \in S_U$, such that no element of $R$
points to a node on the right spine of $U$. Hence, Invariant~\ref{inv_3} holds.
Consequently, the stack $S_U$ satisfies all the invariants after its combination with $S_V$.
See Figure~\ref{fig:sequence_of_joins}.


\begin{figure}
	\centering
	\subfigure[The sequence of trees.]{\includegraphics[scale = 1]{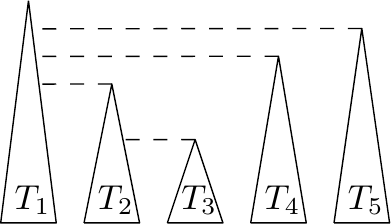}}
	\hfil
	\subfigure[The elements of the stack $S_1$ during the sequence of the join operations.]
	{\includegraphics[scale = 1]{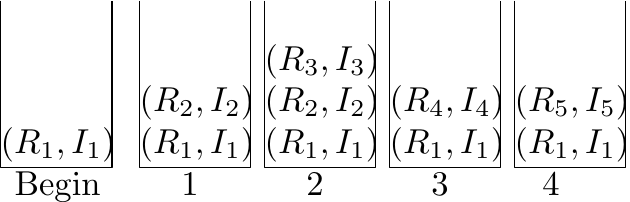}}
	\caption{We join $T_1$ with $T_2, T_3, T_4, T_5$.
		Each $T_i$ has stack $S_i$ with a pair $(R_i, I_i), i = 1 \dots 5$.
		We add $(R_2, I_2)$
		to $S_1$ during the join of $T_1$ and $T_2$.
		We do the same during the join of $T_1$ and $T_3$.
		Next, we pop $(R_3, I_3), (R_2, I_2)$ from $S_1$
		and add $(R_4, I_4)$ to it during the join of $T_1$ and $T_4$.
		We pop $(R_4, I_4)$ from $S_1$ and
		add $(R_5, I_5)$ to it during the join of $T_1$ and $T_5$.}
	\label{fig:sequence_of_joins}
\end{figure}

\paragraph{Maintaining the Invariants in Parallel}
Here we present a parallel algorithm to maintain the Invariants~\ref{inv_2} and~\ref{inv_3} 
throughout all join operations.
More precisely,
we show how to perform a sequence of pop operations on a stack and a combination of
two stacks in constant time on a CREW PRAM.
We demand that each element of the stack has a dedicated PE for this.
Initially, a stack $S$ of a tree contains only one element
and we dedicate it the PE that corresponds to this tree.
Each element $e$ of $S$ contains additional pointers: $\mathrm{Start}$, $\mathrm{StackID}$, and $\mathrm{Update}$.
$\mathrm{Start}$ points to the flag that signals to start the parallel pop operations.
$\mathrm{StackID}$ points to the unique id number of the stack containing the element $e$.
$\mathrm{Update}$ is used to update
$\mathrm{Start}$ and $\mathrm{StackID}$ pointers in new elements of the combination of two stacks.
Each element of a stack has the same $\mathrm{Start}$, $\mathrm{Update}$
and $\mathrm{StackID}$. We refer to this condition as the \textit{stack invariant}.
Additionally, we maintain a global array $\mathrm{UpdateData}$ of size $p$.

Here we discuss how to combine stacks $S_U$ and $S_V$ and maintain the invariants in constant time.
The combination can be done 
in constant time, since the stacks are implemented as linked lists.
Next, we repair the stack invariant; we set the values of the pointers 
$\mathrm{Start}$, $\mathrm{Update}$, $\mathrm{StackID}$ in the old elements of $S_U$
to the values of the pointers in the new elements of $S_U$.
The PE dedicated to tree $U$ starts the repair using
Algorithm~\ref{code:combine_start_repair}. Each PE dedicated to the element $s \in S_U$
that was in $S_V$ permanently performs Algorithm~\ref{code:combine_repair} until
it updates the variables $\mathrm{Start}, \mathrm{Update}, \mathrm{StackID}$ in $s$.
Finally, we wait a constant amount of time until
each PE finishes%
\iftoggle{long}{%
	~updating the pointers of its corresponding element.%
}{%
.
}

\begin{algorithm}[h]
	\caption{Starts the repair of the stack invariant of $S_U$ after combination of $S_U$ and $S_V$.}
	\SetAlgoNoLine
	\KwIn{old element $s \in S_U$, new element $s' \in S_U$}
	\SetKwFunction{algo}{Start\_Repair\_Stack\_Inv}
	\SetKwProg{myalg}{Function}{}{}
	\myalg{\algo{$s$, $s'$}} {
		$\mathrm{UpdateData[s'.StackID].Start} = \mathrm{s.Start}$\;
		$\mathrm{UpdateData[s'.StackID].Update} = \mathrm{s.Update}$\;
		$\mathrm{UpdateData[s'.StackID].StackID} = \mathrm{s.StackID}$\;
		Set flag pointed to by $\mathrm{s'.Update}$\;
	}
	\label{code:combine_start_repair}
\end{algorithm}
\vspace{0.5cm}
\begin{algorithm}[h]
	\caption{Updates $\mathrm{Start}, \mathrm{Update}, \mathrm{StackID}$
		in the new elements of $S_U$.}
	\SetAlgoNoLine
	\KwIn{new element $s' \in S_U$}
	\SetKwFunction{algo}{Repair\_Stack\_Inv}
	\SetKwProg{myalg}{Function}{}{}
	\myalg{\algo{$\mathrm{s'}$}} {
		\If{$\mathrm{flag\ pointed\ to\ by\ s'.Update\ is\ set}$} {	
			$\mathrm{s'.Start} = \mathrm{UpdateData[s'.StackID].Start}$\;
			$\mathrm{s'.Update} = \mathrm{UpdateData[s'.StackID].Update}$\;
			$\mathrm{s'.StackID} = \mathrm{UpdateData[s'.StackID].StackID}$\;
		}
	}
	\label{code:combine_repair}
\end{algorithm}

Now we present a parallel algorithm that performs a sequence of pop operations on the stack $S_U$.
We can perform any number of the pop operations on a
stack in constant time in parallel,
since each element of the stack has a dedicated PE. Recall that we want to find an element $(R, I)$
in $S_U$ such that $r(V) \in I$.
The PE dedicated to the tree $U$ sets the 
flag pointed to by $\mathrm{Start}$ to start a sequence of pop operations.
Each PE that is dedicated to an element $s = (R, I) \in S_U$
permanently performs the function \texttt{Pop(s)}.
This function checks the flag pointed to by $\mathrm{Start}$; it exits if
the flag is unset.
Otherwise, the function checks if $r(V) \in I$;
if $r(V)$ is to the right of $I$ on the integer axis then we mark the element $s$ as deleted.
Next, we perform another parallel test.
Each PE -- if its corresponding element
is not marked deleted -- checks if the next element of the stack is deleted. 
There is only one element that is not marked deleted, but the next element is marked deleted, since
the intervals of elements of $S_U$ are sorted according to Invariant~\ref{inv_1}.
This element is a pair $(R, I)$ such that $r(V) \in I$. 
Next, we make a parallel deletion of the marked elements
and wait a constant amount of time until each PE finishes.

\subsubsection{Splitting a Sequence of Degree~$b$ Nodes in Parallel}
\label{subsubsec:split_b_degree_seq}
Here we present a modified join operation of two trees $U$ and $V$,
where $r(U) \ge r(V)$ (the case  $r(U) < r(V)$ is similar), that performs all splits
in constant time in parallel. We demand the following preconditions:
\nottoggle{long}{%
	(1) we join a sequence of preprocessed trees; that is,
	there are no nodes of degree $b$ on the right spine in the beginning;
	(2) we know a pointer to the node $n \in U$ where $r(V) = r(n)$.
}{%
\begin{enumerate}
	\item We join a sequence of preprocessed trees; that is,
	there are no nodes of degree $b$ on the right spine in the beginning.
	
	\item We know a pointer to the node $n \in U$ where $r(V) = r(n)$.
\end{enumerate}
}%

The algorithm first works as the basic join operation from Section~\ref{sec:preliminaries}.
We can perform a sequence of splits of degree-$b$ nodes in parallel using the fact that there is a dedicated PE
to each node with degree $b$.
This is a case because we assign the PE previously responsible for handling 
tree $V$ to a node $v \in U$ after the join operation, if the degree of $v$ becomes equal to $b$.

\iftoggle{long}{%
	Let us prove that each join operation can
	increase the length of a sequence of degree-$b$ nodes by at most one.
	This fact allows us to assign a dedicated PE to each degree~$b$ node.
}{}%
\begin{lemma} \label{lemma:path_growth}
	A join operation can be implemented so that sequences of degree-$b$ nodes
	grow by at most one element.
\end{lemma}

\iftoggle{long}{%
	\begin{proof}    
		First, we analyze the case when a degree-$b$ node appears during a join operation.
		Let node $n$ has degree $b - 1$ and $p(n)$ is in the sequence of degree-$b$ nodes.
		Figure~\ref{fig:grow_sec_a} shows that if the rightmost child of $n$ splits 
		then we insert a splitter key into $n$ and its degree is equal to $b$ now.
		We show that the sequence of degree-$b$ nodes (where $p(n)$ is)
		grows only by an one node by proving the following fact:
		the new child of $n$ has degree less than $b$ after the join
		operation. If the rank of the new child is greater than $r(V)$ 
		then it has degree less than $b$, since it has been split.
		The case when the rank of the new child equal to $r(V)$ we further analyze.
		
		Consider two sequences of degree-$b$ nodes in $U$ and $V$ and a
		node $n' \in U$ on the right spine such that $r(n') = r(V)$.
		We show that they can not be combined into one sequence after the join operation. 
		We consider only the case when the sequence in $V$ contains its root.
		Otherwise it is obvious that the sequences will not be combined. 
		
		Consider the case when $p(p(n'))$ is in a sequence of degree-$b$ nodes and
		$p(n')$ has degree $b - 1$.
		If the fuse of $n'$ and the root of $V$ does not occur then
		the root of $V$ will be a new child of $p(n')$. See Figure~\ref{fig:grow_sec_b}.
		The length of the sequence
		will increase by one, since the degree of $p(n')$ will be $b$.
		If the root
		of $V$ has degree $b$ then the sequences of degree-$b$ nodes will be combined.
		Hence, our goal to ensure that the degree of the root of $V$ remains less than $b$.
		Thus, we split the root of the tree $V$ if it has degree $b$ before the join operation
		and increase the height of $V$ by one.
		
		Consider the case when $p(n')$ is in a sequence of degree-$b$ nodes.
		The fuse of the root of $V$ and $n'$ occurs when 
		at least one of them
		has degree less than $a$. We fuse them into $n'$ that may result
		in $n'$ having degree $b$. See Figure~\ref{fig:grow_sec_c}.
		Consequently, $n'$ has degree $b$ as well as its
		parent and the sequences will be combined. 
		Then we split $n'$ and insert the splitter key in $p(n')$.
		This split prevents the combination of two sequences.
		We also split the nodes in the sequence of degree-$b$ nodes where
		$p(n')$ is, since the degree of $p(n')$ is $b + 1$.
		We do this in parallel in constant time as further described.
	\end{proof}
	
	\begin{figure}[t]
		\centering
		\subfigure[]{\includegraphics[scale = 0.8]{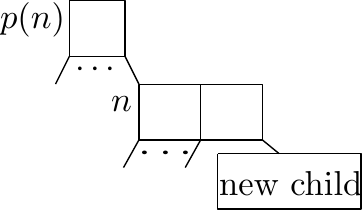}\label{fig:grow_sec_a}}
		\hfil
		\subfigure[]{\includegraphics[scale = 0.8]{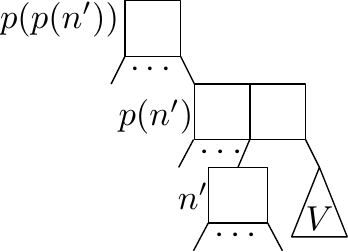}\label{fig:grow_sec_b}}
		\hfil
		\subfigure[]{\includegraphics[scale = 0.8]{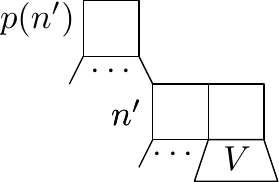}\label{fig:grow_sec_c}}
		\caption{The growth of a sequence of degree-$b$ nodes.}
		\label{fig:grow_sec}
	\end{figure}
}{}%

\paragraph{Assigning PEs} 
We now explain how to assign a PE to a new degree-$b$ node $n$. 
Suppose that $n$ has extended some sequence of degree-$b$ nodes according to 
Lemma~\ref{lemma:path_growth}. We assign the freed PE of the tree $V$ to $n$.
Hence, this PE can split $n$ during some following join operation.

Now let us discuss how we assign a PE to $n$ in more detail.
We extend each node by an additional pointer $\mathrm{MetaData}$ to a special data structure:
a flag $\mathrm{Start}$ and an integer $\mathrm{Rank}$.
The $\mathrm{Start}$ flag signals to all PEs
assigned to the nodes of the same sequence of degree-$b$ nodes to start the parallel split of
this sequence. 
We assign the pointer $\mathrm{MetaData}$ in $p(n)$ to the pointer $\mathrm{MetaData}$ in $n$.
Finally, we command to the PE that is dedicated to $n$ to perform
permanently the function \texttt{Split\_B\_Node($n$)}, which is described
further.

\paragraph{Splitting a Sequence of Degree-$b$ Nodes}
The PE dedicated to $U$ assigns $r(V)$ to $\mathrm{Rank}$ and next
sets $\mathrm{Start}$ when we need to perform a split of nodes of a sequence
of degree-$b$ nodes in parallel. 
Each PE dedicated to a node $n$
in this sequence permanently performs the function \texttt{Split\_B\_Node($n$)}. 
This function checks the flag $\mathrm{Start}$; it exits if the flag is unset.
Otherwise, the function starts splitting $n$ if $r(n) \ge \mathrm{Rank}$.

Let us consider the parallel split of the sequence of degree-$b$ nodes more precisely.
Each PE splits its corresponding node $y$ by creating a new node $x$ and coping the first $\lfloor\frac{b}{2}\rfloor - 1$
keys from $y$ to $x$. The last $\lceil\frac{b}{2}\rceil - 1$ keys remain in $y$.
Next, the PE waits until $p(y)$ is split and then inserts the splitter key
and the pointer to $x$ into $p(y)$.
Next, we 
wait
until each PE finishes%
\iftoggle{long}{%
	.
}{%
~its corresponding task.
}%
This takes constant time on a PRAM.
\iftoggle{long}{%
	See Figure~\ref{fig:split_sequence}.
}{}%
It is
crucial that all nodes with the degree $b$ sequence
are still the parents of their rightmost children after the split,
because we know only the
parent pointer for each child (we store a parent pointer in each node). 
Note that all the nodes which were on the right spine before the split step remain
on the right spine after it. This property is crucial to access spine nodes in constant
time.

\iftoggle{long}{%
	\begin{figure}[t]
		\centering
		\subfigure[]{\includegraphics[scale = 0.8]{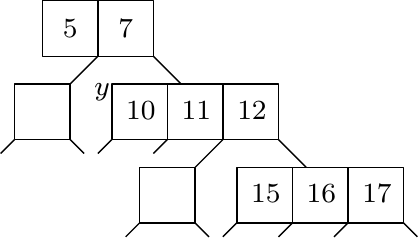}\label{fig:split_sequence_a}}
		\hfil
		\subfigure[]{\includegraphics[scale = 0.8]{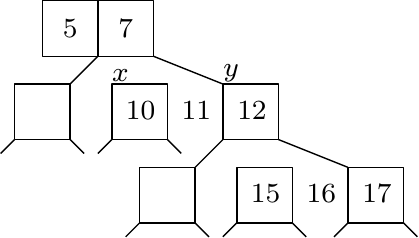}\label{fig:split_sequence_b}}
		\hfil
		\subfigure[]{\includegraphics[scale = 0.8]{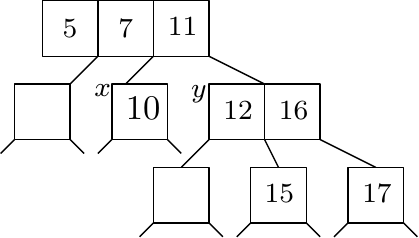}\label{fig:split_sequence_c}}
		\caption{We split a sequence of the two $4$-degree nodes in $(2, 4)$-tree \protect\subref{fig:split_sequence_a}.
			First, we split each of the nodes in parallel \protect\subref{fig:split_sequence_b}.
			The nodes with keys $7$ and $12$ remain parents of their rightmost child.
			Next, we insert the splitter keys $11$, $16$ and pointers to the nodes with keys $10$ and $15$
			into the parent nodes \protect\subref{fig:split_sequence_c}.}
		\label{fig:split_sequence}
	\end{figure}
}{}%

\subsubsection{Maintaining Subtree Sizes}

\label{subsubsec:subtree_sizes}
Here we present a parallel algorithm to update the size of each subtree in a tree $T$, where $T$ is the result
of joining $T_1, \cdots, T_k$.
First, we save all the nodes where
the joins of two trees occurred during the join operation of $k$ trees.
Suppose we join two trees $U$ and $V$, such that
$r(U) \ge r(V)$. We save a node $n \in U$ on the right spine that has rank $r(n) = r(V)$. Next,
we dedicate the PE that was previously dedicated to $V$ to node $n$.
Consequently, each of these nodes has a unique dedicated PE.

Now suppose we have finished joining $k$ trees. 
Consider $k - 1$ saved nodes and the PEs dedicated to them.
The PE that made the last join operation of two trees sets the global
flag $\mathrm{Join\_Done}$ and performs the function $\texttt{Update\_Subtree\_Sizes}$.
Other PEs permanently perform the function $\texttt{Update\_Subtrees\_Sizes}$ as well.
This function checks the flag $\mathrm{Join\_Done}$; it exits if the flag is unset.
Otherwise, it updates the subtree sizes of $T$ as follows: each PE follows up the path from 
corresponding saved node to the root of $T$ and updates the subtree sizes.
The algorithm works in $\mathcal{O}(\log m)$ time on a CREW PRAM.


\subsubsection{The Parallel Join Algorithm}
\label{subsubsec:join_algorithm}
Now we have presented the necessary subroutines and can use them to construct 
the parallel join algorithm.

\begin{lemma} \label{lemma:const_cost_par}
	We can join trees $U$ and $V$, where $r(U) \ge r(V)$ (the case
	$r(U) < r(V)$ is similar), in constant time.
\end{lemma}

\iftoggle{long}{%
	\begin{proof}
		We retrieve a node $u$ with rank $r(V)$ on the right spine of $U$ in
		constant time as described in Section~\ref{subsubsec:node_retrieve_par}.
		Next, we insert the root of $V$ as the rightmost child
		of $p(u)$ or fuse it with $u$.
		Finally, we perform the splits of nodes with degree $\geq b$
		in constant time as described in~\ref{subsubsec:split_b_degree_seq}.
		Therefore, we join $U$ and $V$ in constant time.
	\end{proof}
}{}%

We have $T_1, \dots, T_k$ trees and each of these trees has its own dedicated PE. 
Each PE $i$, where $i$ is odd, performs the modified join operation of $T_i$ and $T_{i + 1}$
in constant time according to Lemma~\ref{lemma:const_cost_par}.
If $r(T_i) \ge r(T_{i + 1})$ then PE $i + 1$ is freed
after this join operation, otherwise PE $i$ is freed.
Now the freed PE performs Algorithm~\ref{code:main}
and the other PE performs the next join operation.
Finally, each PE deletes the corresponding stack, left and right spines
when all trees are joined. Therefore, Lemma~\ref{lemma:nonoptimal_parallel_join} holds.

\iftoggle{long}{%
	\begin{algorithm}[h]
		\caption{Main function, which is performed by all freed PEs}
		\SetAlgoNoLine
		\KwIn{a node $n$ of degree $b$, a stack element $s$}
		\SetKwFunction{algo}{Main}
		\SetKwFunction{Split}{Split\_B\_Node}
		\SetKwFunction{Repair}{Repair\_Stack\_Invariant}
		\SetKwFunction{Pop}{Pop}
		\SetKwFunction{UpdateSizes}{Update\_Subtree\_Sizes}
		\SetKwProg{myalg}{Function}{}{}
		\myalg{\algo{$\mathrm{n}, \mathrm{s}$}} {
			\Split{n}\;
			\Repair{s}\;
			\Pop{s}\;
			\UpdateSizes{}\;
		}
		\label{code:main}
	\end{algorithm}
	\setlength{\intextsep}{0pt} 
	\setlength{\textfloatsep}{0pt}
}{%
\begin{algorithm}[h]
	\caption{Main function, which is performed by all freed PEs}
	\SetAlgoNoLine
	\KwIn{a node $n$ of degree $b$, a stack element $s$}
	\SetKwFunction{algo}{Main}
	\SetKwFunction{Split}{Split\_B\_Node}\SetKwFunction{Pop}{Pop}
	\SetKwFunction{Repair}{Repair\_Stack\_Invariant}
	\SetKwFunction{Pop}{Pop}
	\SetKwFunction{UpdateSizes}{Update\_Subtree\_Sizes}
	\SetKwProg{myalg}{Function}{}{}
	\myalg{\algo{$\mathrm{n}, \mathrm{s}$}} {
		\Split{n}\;
		\Repair{s}\;
		\Pop{s}\;
		\UpdateSizes{}\;
	}
	\label{code:main}
\end{algorithm}
\setlength{\intextsep}{0pt} 
\setlength{\textfloatsep}{0pt}
}%

\subsection{Sequential join of $t$ trees}\label{subsec:seq_join_trees}
We present a sequential algorithm to join $t$ preprocessed trees in $\Ohsmall{t}$ time.
This algorithm joins trees in pairs.
During a join operation
we, first, access a spine node by its rank;
next, we connect the trees and split nodes of degree~$b$.
Both operations can be done in amortized constant time.

\paragraph{Fast Access to Spine Nodes by Rank} \label{subsec:seq_access_node}
Consider joining $U$ and $V$, where $r(U) \ge r(V)$ (the case
$r(U) < r(V)$ is similar).
We use the same idea as in Section~\ref{subsubsec:node_retrieve_par}
to retrieve a spine node with rank $r(V)$ in $U$.
We maintain a stack with arrays of the pointers to the nodes on the right and left spines for each tree
that we have built during the preprocessing step. Initially, each stack contains one element.
Each tree and its corresponding stack satisfy the invariants from Section~\ref{subsubsec:node_retrieve_par}%
\iftoggle{long}{%
	that guarantees that the pointers of the arrays in the stack point to the nodes on the right (left) spine.
}{%
.
}%
To maintain the invariants we, first, perform a sequence of pop operations on the stack of $U$ to retrieve a spine node; next, we combine stacks of $U$ and $V$.
See details in Section~\ref{subsubsec:node_retrieve_par}.

We combine two stacks in worst-case constant time, since each stack
is represented using a linked list. 
Each pop operation removes an element that was in the stack as a result of the
previous combine operation.
Therefore, we can charge the cost of the sequence of $s$ pop operations
to the previous $s$ combine operations.
Hence, the sequence of $s$ pop operations takes amortized
constant time.
See the detailed proof of the amortized constant cost of the sequence of $s$ pop operations
in~\cite[Chapter~17: Amortized Analysis,~p. 460 -- 461]{Cormen:2009:IAT:1614191}.

\begin{lemma} \label{lemma:seq_access_node}
	Consider joining of two trees.
	We can access a spine node by its rank in amortized constant time.
\end{lemma}

\begin{lemma} \label{lemma:seq_split_b_nodes}
	Consider joining of $t$ preprocessed trees $T_1, \dots, T_t$.
	We split $\Ohsmall{t}$ degree-$b$ nodes over all operations.
\end{lemma}

\begin{proof}
	We use the potential
	method~\cite{Cormen:2009:IAT:1614191} to prove this fact.
	method~\cite{Cormen:2009:IAT:1614191}
	We define the potential function
	$\phi$ as the total number of the nodes of degree~$b$ on the left and rights spines of trees to be 
	joined.
	Suppose that the sequential join algorithm has joined trees $T_1, \dots, T_{i - 1}$. We denote the 
	result
	as $T$ and let $\phi(T) = \phi_{i - 1}$.
	Then the amortized number of splits that occurred during the join operation
	of $T$ and $T_i$ is $\hat{c_i} = c_i + \phi_i - \phi_{i - 1}$,
	where $c_i$ is the actual number of occurred splits.
	Note that $\phi_{i - 1} + 1 \ge \phi_i$ and $c_i \le |\phi_i - \phi_{i - 1}|$,
	therefore $\hat{c_i} = \Oh{1}$ and
	$\sum_{i = 2}^{t} \hat{c_i} = \sum_{i = 2}^{t} c_i + \phi_t - \phi_1$.
	Initially, the trees are preprocessed and do not contain nodes of degree~$b$ on the spines,
	hence $\phi_1 = 0$
	and $\sum_{i = 2}^{t} c_i \le \sum_{i = 2}^{t} \hat{c_i} = \Ohsmall{t}$.
\end{proof}

\begin{lemma} \label{lemma:seq_join_trees}
	We can join $t$ trees in $\Ohsmall{t}$ time.
\end{lemma}
\iftoggle{long}{%
	\begin{proof}
		Consider the join of two trees $U$ and $V$ where $r(U) \ge r(V)$.
		First, we search for a spine node in $U$ to
		insert the root of $V$ in amortized constant time according
		to Lemma~\ref{lemma:seq_access_node}. 
		Next, we split the nodes of degree~$b$ in the resulted tree in amortized
		constant time according to Lemma~\ref{lemma:seq_split_b_nodes}.
		Therefore, we join $t$ trees in $\Ohsmall{t}$ time.
	\end{proof}%
}{%
}%

\subsection{Optimal parallel join}
Now we present an optimal join algorithm with $\Ohsmall{k}$ work and $\Ohsmall{\log k + \log m}$ parallel time.
First, we preprocess the $k$ trees using $k$ processors with $\Ohsmall{k \log \frac{m}{k}}$ work and $\Ohsmall{\log m}$ parallel time.
Next, we split the sequence of trees into the groups of size $\log k$ and join
each group in $\Ohsmall{\log k}$ parallel time
by Lemma~\ref{lemma:seq_join_trees}. The work of this step is 
$\lceil k/\log k \rceil \Ohsmall{\log k} = \Ohsmall{k}$.
We preprocess the $\lceil k/\log k\rceil$ resulting trees
with $\Ohsmall{\lceil k/\log k \rceil \log \frac{m \log k}{k}}$ work
and $\Ohsmall{\log m}$ parallel time.
Next, we join the $\lceil k/\log k\rceil$ trees
using a non-optimal parallel join algorithm with work 
$\Ohsmall{\lceil k/\log k \rceil \log \lceil k/\log k \rceil}$
and $\Ohsmall{\log \lceil k/\log k \rceil}$ parallel time
by Lemma~\ref{lemma:nonoptimal_parallel_join}.
The total work of the algorithm is $\Ohsmall{k \log \frac{m}{k}}$,
the parallel time is $\Ohsmall{\log k + \log m}$,
and Theorem~\ref{th:optimal_parallel_join} holds.
Note that this algorithm can not maintain subtree sizes.

\section{Lightweight Parallel Join} \label{sec:lightweight_join}
The parallel join algorithm
from Section~\ref{sec:parallel_join} is optimal on a CREW PRAM. Because this algorithm is theoretical 
and difficult to implement, we suggest an other approach to join $k$ trees $T_1, \dots, T_k$. 
We devote the rest of this section to \frage{weakened claim}outlining
a proof of the following theorem.  \frage{new sentence. check. fill in citation}The idea is to replace the
pipelining tricks used in previous algorithms \cite{paul1983parallel} by a local
synchronization that can actually be implemented on asynchronous
shared memory machines.

\begin{theorem} \label{th:lightweight_parallel_join}
	We can join $k$ trees with expected $\Ohsmall{k \log \frac{m}{k}}$ work and 
	expected time $\Ohsmall{\log m + \log k}$ 
	using $p = k$ processors on a CREW PRAM.
\end{theorem}

We decrease the running time of the parallel join 
that joins $k$ trees in $\Ohsmall{\log m \log k}$ time
(see Section~\ref{sec:parallel_join}) by
using arrays with pointers to right (left) spine nodes. We build such arrays
during the preprocessing step for each tree. See details in Section~\ref{subsec:non_opt_join}.

First, we assign a PE $t$ to tree $T_t$.
Next, our algorithm works in iterations. We define the sequence of trees
present during iteration $i$ as $T^i_1, \dots, T^i_{k_i}$ ($k_1 = k$).
In the beginning of each iteration we generate a random bit $c^i_t$ 
for tree $T^i_t$ where $t = 1 \dots k_i$.
During an iteration $i$ PE $t$ joins $T^i_t$ to $T^i_{t - 1}$ (or $T^i_t$ and $T^i_{t + 1}$ if $t = 1$),
if one of the following conditions holds:
\iftoggle{long}{%
	\begin{enumerate}
		\item $r(T^i_{t - 1}) > r(T^i_t)$ and $r(T^i_t) < r(T^i_{t + 1})$
		\item $r(T^i_{t - 1}) > r(T^i_t)$, $r(T^i_t) = r(T^i_{t + 1})$,
		and $c^i_t = 1$
		\item $r(T^i_{t - 1}) = r(T^i_t)$, $r(T^i_t) < r(T^i_{t + 1})$,
		$c^i_{t - 1} = 0$, and $c^i_t = 1$
		\item $r(T^i_{t - 1}) = r(T^i_t)$, $r(T^i_t) = r(T^i_{t + 1})$, $c^i_{t - 1} = 0$, and $c^i_t = 1$
	\end{enumerate}
}{%
(1) $r(T^i_{t - 1}) > r(T^i_t)$ and $r(T^i_t) < r(T^i_{t + 1})$;
(2) $r(T^i_{t - 1}) > r(T^i_t)$, $r(T^i_t) = r(T^i_{t + 1})$,
and $c^i_t = 1$;
(3) $r(T^i_{t - 1}) = r(T^i_t)$, $r(T^i_t) < r(T^i_{t + 1})$,
$c^i_{t - 1} = 0$, and $c^i_t = 1$;
(4) $r(T^i_{t - 1}) = r(T^i_t)$, $r(T^i_t) = r(T^i_{t + 1})$, $c^i_{t - 1} = 0$, and $c^i_t = 1$;
}%
These rules ensure that only trees at locally minimal height
are joined and that ties are broken randomly and in such a way that no chains of join operations occur in a single step.
In the beginning, the join operation proceeds like in basic join
operation~\ref{paragraph:basic_join}.  But we do not insert a new
child into the $p(n)$ ($n \in T^i_{t - 1}$ and $r(n) = r(T^i_t)$) if
its degree equals $b$; instead we take $n$ and the new child, join
them, and put the result into $T^i_t$.  We do this in constant time,
since the ranks of $T^i_t$ and the new child are equal. The rank of
$T^i_t$ increases by one. We call this procedure \emph{subtree
	stealing}. This avoids chains of splitting operations that would lead to 
non-constant work in iteration $i$.

\begin{lemma} \label{lemma:ascending_heights}
	Assume that we joined two trees $T^i_{t_i}$ and $T^j_{t_j}$ ($i < j$) with a tree $T$,
	where no keys in $T$ is larger than any key in $T^i_{t_i}$ and $T^j_{t_j}$,
	and
	no joins with this tree occurred between the iterations $i$ and $j$. Then
	$r(T^i_{t_i}) \le r(T^j_{t_j})$.
\end{lemma}

\begin{proof}
	Since $T^i_{t_i}$ and $T$ were joined
	then $r(T) \ge r(T^i_{t_i})$ and $r(T^i_{t_i}) \le r(T^i_{t_i + 1})$.
	But $r(T^j_{t_j}) \ge r(T^i_{t_i + 1})$. Hence, $r(T^j_{i_j}) \ge r(T^i_{t_i})$.
\end{proof}

Lemma~\ref{lemma:ascending_heights} shows that we join trees in ascending order of their ranks.
Therefore, we do not need to update the pointers that point to the nodes with rank
less than $r(T^i_{t_i})$, since $r(T^j_{t_j}) \ge r(T^i_{t_i})$; 
that is, we do not join trees smaller then $r(T^i_{t_i})$.

\begin{lemma} \label{lemma:lightweight_parallel_join}
	The lightweight parallel join algorithm joins $k$ trees using expected $\Ohsmall{\log m + \log k}$ iterations.
\end{lemma}

\begin{proof}
	Consider a tree $T^i_t$ where $r(T^i_{t - 1}) > r(T^i_t)$ and $r(T^i_t) < r(T^i_{t + 1})$.
	Subtree stealing from tree $T^i_t$ may occur $\Ohsmall{\log m}$ times over all iterations.
	Since
	the heights of the trees that steal from $T^i_t$ are in ascending order (according
	to Lemma~\ref{lemma:ascending_heights}) and increment by one after each stealing.
	
	Consider a sequence of $l$ trees with heights in ascending order. The PE dedicated to the smallest tree 
	joins this sequence performing $l$ iterations, but the total length of such
	sequences over all iterations is $\Ohsmall{\log m}$.
	
	Consider the ``plains'' of subsequent trees with equal height. Let 
	$l$ denote the sum of the plain sizes.
	A tree $t$ in a plain is joined to its left neighbor with probability
	at least $1/4$ (if its $c^i_t$ is $1$ and that of its left neighbor is $0$).
	Hence, the expected number of joins on plains is $l/4$. This means that, in expectation,
	the total plain size shrinks by a factor $3/4$ in each iteration. 
	Overall, $\Oh{\log k}$ iterations suffice to remove all plains in expectation. 
	To see that fringe cases are no problem note that also trees at the border of 
	a plain are joined with probability at least $1/4$, possibly higher since
	a tree at the left border of a plain is joined with probability $1/2$.
	Furthermore, other join operations may merge plain but they never increase the 
	number of trees in plains.
\end{proof}

Combining the results of Lemma~\ref{lemma:lightweight_parallel_join} and the fact that we do not need
to update pointers to right (left) spines we prove that the running time of this algorithm is 
$\Ohsmall{\log m + \log k}$ in expectation. 
The work of the algorithm is $\Ohsmall{k \log \frac{m}{k}}$ in expectation, since
the PE $t$ ($t = 2, \dots, k - 1$) performs the number of iteration that is proportional 
to $\max (r(T^1_{t - 1}), r(T^1_{t + 1}))$.
This proves Theorem~\ref{th:lightweight_parallel_join}.

\section{Bulk Updates} \label{sec:bulk_updates}
We present a parallel data structure with bulk updates. First, we describe
a basic concepts
behind our data structure. Suppose we have an $(a, b)$-tree $T$ and a sorted sequence
$I$ for the bulk update 
of $T$. Our parallel bulk update algorithm is based on the idea presented in
\cite{FriasS07}
and consists of the three phases: \textbf{split}, \textbf{insert/union} and
\textbf{join}.

First, we choose a sorted sequence of separators $S$ from $I$ and $T$.
Next, the \textbf{split} phase splits $T$ into $p$ trees using separators $S$.
Afterwards, the \textbf{insert/union} phase inserts/unions elements of each subsequence of $I$
to/with the corresponding tree. Finally, the \textbf{join} phase joins the trees back into a tree.
The following theorem summarizes the results of this section.

\begin{theorem}\label{thm:bulkUpdate}
	A bulk update can be implemented to run in  
	$\Ohsmall{\frac{|I|}{p}\log \frac{|T|}{|I|} + \log p + \log |T|}$ parallel time on $p$ PEs of
	a CREW PRAM.
\end{theorem}

The algorithm needs $\Ohsmall{|I|\log \frac{|T|}{|I|}}$ work
and $\Ohsmall{\log |T|}$ time using $p = |I|\log\frac{|T|}{|I|} / \log |T|$ on a CREW PRAM.
Note that our algorithm is optimal, since sequential
union of $T$ and $I$ requires $\Omega(|I|\log
\frac{|T|}{|I|})$ time in the worst case. See~\cite{DBLP:journals/acta/HuddlestonM82}.

\subsection{Selecting the Separators.}\label{sec:selection}
The complexity of the algorithm depends on how we choose the separators. 
Once we have selected the separators we can split sequence $I$ and tree $T$
according to them.
Depending on the selection of the separators we will perform either the insert or union phase.

In \textit{Uniform Selection}, we select $p - 1$ separators that 
split $I$ into $p$ disjoint equal-sized subsequences $\mathcal{I} = \{ I_1, I_2, \dots, I_p \}$
where $\bigcup_{i = 1}^{p} I_i = I$ (we assume that $|I|$ is divisible by $p$ for simplicity).

\paragraph{Selection with Double Binary Search}
We adapt a technique used in parallel merge algorithms
\iftoggle{long}{%
	\cite{journals/jal/ShiloachV81, Jaja_Parallel_Algs}
}{%
\cite{journals/jal/ShiloachV81}
}%
for our purpose.
First, we select $p - 1$
separators $S = \langle s_1, \dots, s_{p - 1} \rangle$ that split $I$ into $p$ equal
parts. Next, we select $p - 1$
separators $T_{\mathrm{sep}} = \langle t_1, \dots, t_{p - 1} \rangle$ 
that divide the sequence represented by the tree $T$ into equal parts.
We store in each node of $T$ the sizes of its subtrees
in order to find these separators in logarithmic time.
See the details in~\cite{Cormen:2009:IAT:1614191}.

\iftoggle{long}{%
	\begin{figure}
		\centering
		\includegraphics[scale = 1.2]{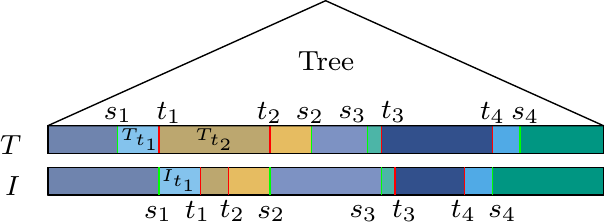}
		\caption{Separators of $I$ and $T$. Here $p = 5$. For example,
			the elements from $I_{t_1}$ and $T_{t_1}$ lay in intersecting ranges. But the elements from
			$T_{t_1}$
			and $T_{t_2}$ lay in disjoint ranges. As well as the elements of $I_{t_1}$ and $T_{t_2}$.}
		\label{fig:double_selection}
	\end{figure}
}{}%

Now consider $x \in S \cup T_{\mathrm{sep}}$. The subsequence $I_x$ of $I$ includes
all elements $\le$ $x$ but greater than $y$, the largest
element in $S \cup T_{\mathrm{sep}}$ that is less than $x$.
Similarly, we define subtree $T_x$.
We implicitly represent the subsequence $I_x$ and subtree $T_x$ by $y$ and $x$.
Each PE uses binary search to find $y$.
For example, PE $i$ calculates the implicit representation of $I_{s_i}, T_{s_i}, I_{t_i}, T_{t_i}$.
These searches take time at most $\mathcal{O}(\log p)$, since $S$ and $T_{\mathrm{sep}}$ are sorted and 
we can use binary search.
Thus, we split $I$ and $T$ into $2p - 1$ parts each:
$\mathcal{I} = \{ I_x : x \in S \cup T_{\mathrm{sep}}\}$ and $\mathcal{T} = \{ T_x : x \in S
\cup T_{\mathrm{sep}}\}$ respectively.

Note that elements of two arbitrary subsequences from $\mathcal{I} \cup \mathcal{T}$
lay in
disjoint ranges, except for $I_x$ and $T_x$ for the same $x \in S \cup T_{\mathrm{sep}}$. Also $|I_x|
\le |I| / p$ and
$|T_x| \le |T| / p$, since the distances between neighboring separators in $S$ and $T$ are at most
$|I| / p$ and $|T| / p$, respectively.
\iftoggle{long}{%
	See Figure~\ref{fig:double_selection}.
}{}%


%

\paragraph{Split and Insert phases}
First, we select $p - 1$ separators using uniform selection.
Next, we split the tree $T$ into $p$ subtrees $T_1, \dots, T_p$ using the parallel split algorithm from
Section~\ref{sec:parallel_split}.
Finally, we insert the subsequence $I_i$ into corresponding subtree $T_i$, for $i = 1 \dots p$,
in $\mathcal{O}(|I|/p \log|T|)$ time on a CREW PRAM.

\paragraph{Split and Union phase}
Suppose we selected the separators using selection with double binary search.
First, we split the sequence $I$ into subsequences $\mathcal{I}$.
Next, we split $T$ into subtrees $\mathcal{T}$.
For each subtree $T_x \in \mathcal{T}$ we split $T$ by the representatives $y$ and $x$ of
$T_x$ using the parallel split algorithm from Section~\ref{sec:parallel_split}.


%

Now each PE $i$ unions $I_{s_i}$ with $T_{s_i}$ and $I_{t_i}$ with $T_{t_i}$ ($s_i, t_i \in S
\cup T_{\mathrm{sep}}$) using the sequential union algorithm from Section~\ref{sec:preliminaries}.
This algorithm unions $I_x \in \mathcal{I}$ and $T_x$ 
in $\mathcal{O}(|I_x| \log\frac{|T_x|}{|I_x|})$ time.
Hence, the union phase can be done in $\mathcal{O}(|I|/p \log\frac{|T|}{|I|})$ time,
because $|I_x| \le |I| / p$ and $|T_x| \le |T| / p$ for any $x \in S \cup T_{\mathrm{sep}}$
and because for $|T_x|\geq e|I_x|$, $|I_x| \log\frac{|T_x|}{|I_x|}$ is maximized for 
the largest allowed value of $|I_x|$. For the remaining cases, we can use that the log term is bounded by a constant anyway.

\paragraph{Join Phase}\label{subsec:join}
We join the trees $T_1, \dots, T_k$ ($k = p$ or $2p - 1$) into the tree $T$ in
$\mathcal{O}(\log p + \log |T|)$ time on a CREW PRAM
using the non-optimal parallel join algorithm from Section~\ref{sec:parallel_join}.
We use the non-optimal version, because it maintains subtree sizes in nodes of trees.

%

\section{Operations on Two Trees}\label{furtherOps}

\begin{theorem}
	Let $U$ and $T$ denote search trees ($|U|\leq |T|$) that we identify with sets of elements. 
	Let $k=\max\{|U|,|T|\}$ and $n=\max\{k, p\}$.
	Search trees for $U\cup T$ (union), $U\cap T$ (intersection), $U\setminus T$ (difference), and $U\bigtriangleup T$ (symmetric difference) can be computed in time $\Ohsmall{k/p\log\frac{n}{k}+\log n}$.
\end{theorem}

\paragraph{Union}
Computing the union of two trees $U$ and $T$ can be implemented by
viewing the smaller tree as a sorted sequence of insertions. 
Then we can extract the elements of $U$ in sorted order using work 
$\Ohsmall{k}$ and span $\Ohsmall{\log n}$.
Adding the complexity of a bulk insertion from Theorem~\ref{thm:bulkUpdate} we get
total time $\Ohsmall{k/p\log \frac{n}{k}+\log n}$.

\paragraph{Intersection}
We perform a bulk search for the elements of $U$ in $T$ in time 
$\Ohsmall{k/p\log \frac{n}{k}+\log n}$ and build a search tree from this sorted sequence in time 
$\Ohsmall{k/p+\log n}$.

\paragraph{Set Difference $T\setminus U$}
If $|T|\geq |U|$ we can interpret $U$ as a sequence of
deletions.
This yields parallel time $\Ohsmall{k/p\log \frac{n}{k}+\log n}$. If $|U|> |T|$, we first
intersect $U$ and $T$
in time $\Ohsmall{k/p\log \frac{n}{k}+\log n}$ and then compute the
set difference $T\setminus U=T\setminus (T\cap U)$ in time
$\Ohsmall{k/p\log \frac{n}{k}+\log n}$.

\paragraph{Symmetric Difference}
We use the results for union and difference and apply the definition
$U\bigtriangleup T=(U\setminus T)\cup(T\setminus U)$.

\section{Experiments}
In this section, we evaluate the performance of our parallel $(a, b)$-trees
and compare them to the several number of contestants.
We also study our basic operations, join and split, in isolation.

\paragraph{Methodology} We implemented our algorithms using \texttt{C++}. Our implementation 
uses the \texttt{C++11} multi-threading library (implemented using the POSIX thread library)
and an allocator \texttt{tbb::scalable\_allocator}
from the Intel Threading-Building Blocks (TBB) library.
All binaries are built using \texttt{g++-4.9.2} with the \texttt{-O3} flag. We run our experiments
on Intel Xeon E5-2650v2 (2 sockets, 8 cores with Hyper-Threading, 32 threads)
running at 2.6 GHz with 128GB RAM. To decrease scheduling effects,
we pin one thread to each core.

Each test can be described by the size of an initial $(4, 8)$-tree ($T$), the size of a bulk update ($B$),
the number of iterations ($I$) and the number of processors ($p$). Each key is a 32-bit integer.
In the beginning of each test we construct an initial tree. During
each iteration we perform an incremental bulk update.
We use a uniform, skewed uniform\frage{todo: explain},
normal and increasing uniform distributions to generate
an initial tree and bulk updates in most of the tests.
The skewed distribution generates keys in smaller range than a uniform distribution.
The increasing uniform distribution generates the keys of a bulk update such that
the keys of the current bulk update are greater than the keys of the previous bulk
updates.

\paragraph{Split algorithms}
\label{paragraph:split_performance}
We present a comparison of the sequential split algorithm
and the parallel split algorithm from Section~\ref{sec:parallel_split}.
We split an initial tree into $31$ trees using $30$ sorted separators.
The sequential algorithm splits the pre-constructed tree into two trees using 
the first separator and the split operation from Section~\ref{paragraph:basic_split}.
Next, it continues to split using the second separator, and so on.
Figure~\ref{fig:split_test} shows the running times of the tests with a uniform distribution
(other distributions have almost the same running times).
The parallel split algorithm outperforms the sequential algorithm by a factor of $8.1$.
Also we run experiments with $p = 16$
and split an initial tree into $16$ trees. The speed-up of the parallel split is $4.6$.

\iftoggle{long}{%
	\begin{figure}[h]
		\centering
		\includegraphics[scale = 0.60]{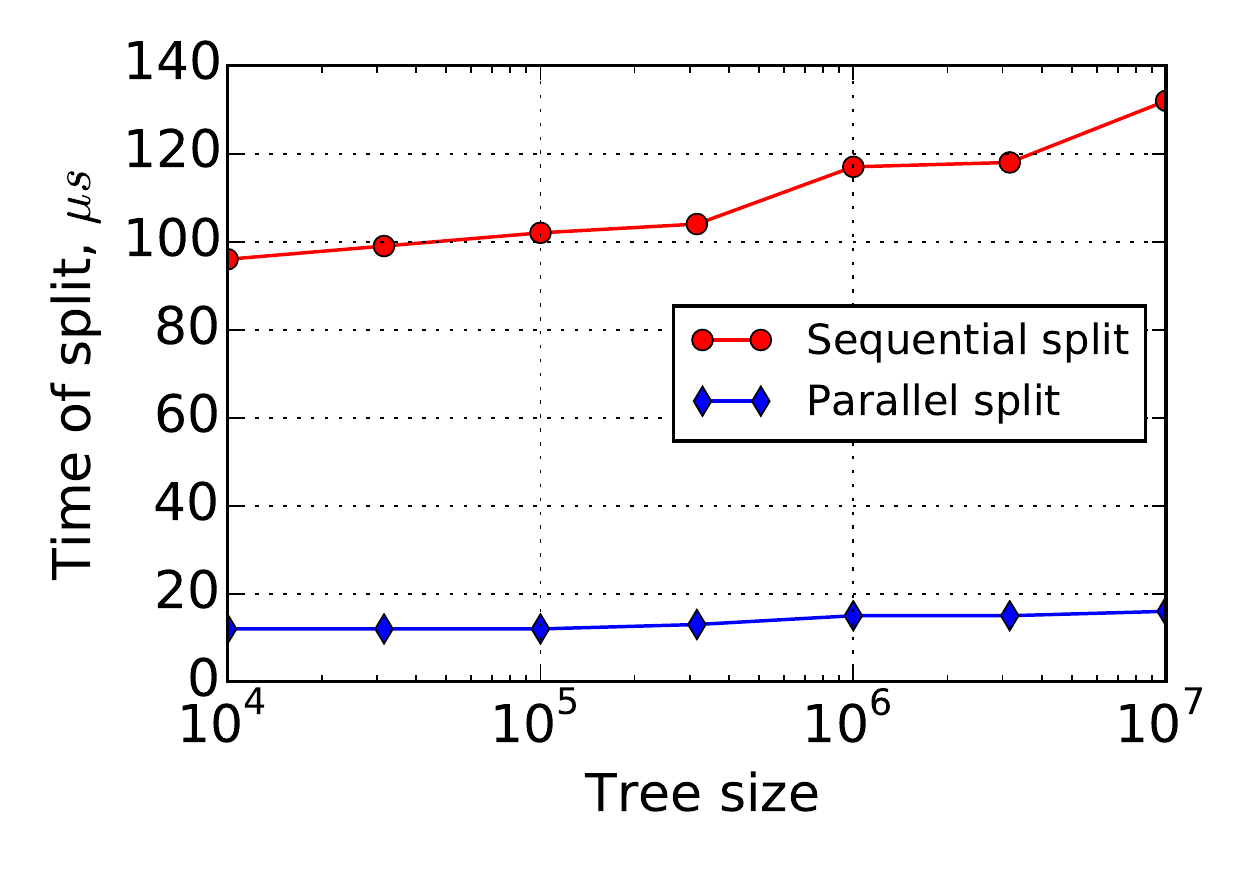}
		\caption{Comparison of split algorithms. $T = \sqrt{10} ^ i, i = 8, \dots, 14$,
			and $p = 31$}
		\label{fig:split_test}
	\end{figure}
}{}%

\paragraph{Join algorithms}
\label{paragraph:join_performance}
We present a comparison of different join algorithms.
We join $31$ trees using the following three algorithms:
(1) \textbf{SJ} is a sequential join algorithm. It joins the first tree
to the second tree using join operation from Section~\ref{paragraph:basic_join}.
Next, it joins the result of previous join and the third tree, and so on.
(2) \textbf{PPJ} is a parallel pairwise join algorithm that joins pairs of trees in parallel from
Section~\ref{sec:parallel_join}.
(3) \textbf{PJ} is a parallel join algorithm from Section~\ref{sec:lightweight_join}.
Figure~\ref{fig:time_uniform} shows the running times of the tests with a uniform distribution
(other distributions have almost the same running times).
The \textbf{PJ} algorithm is worse than the \textbf{SJ} and the \textbf{PPJ} algorithm 
by a factors of $1.9$ and $3.4$, respectively. The \textbf{PPJ} algorithm has a speed-up of $1.8$ 
compared to the \textbf{SJ}.
This can be explained by the involved synchronization overhead.
In the experiments with $p = 16$ and $16$ trees the \textbf{PPJ} and
\textbf{SJ} algorithms show the same running times. We explain this by the fact that 
hyper-threading advantages when there are a lot of dereferences of pointers and the number 
of trees is nearly doubled.

\iftoggle{long}{%
	\begin{figure}[h]
		\centering
		\includegraphics[scale = 0.60]{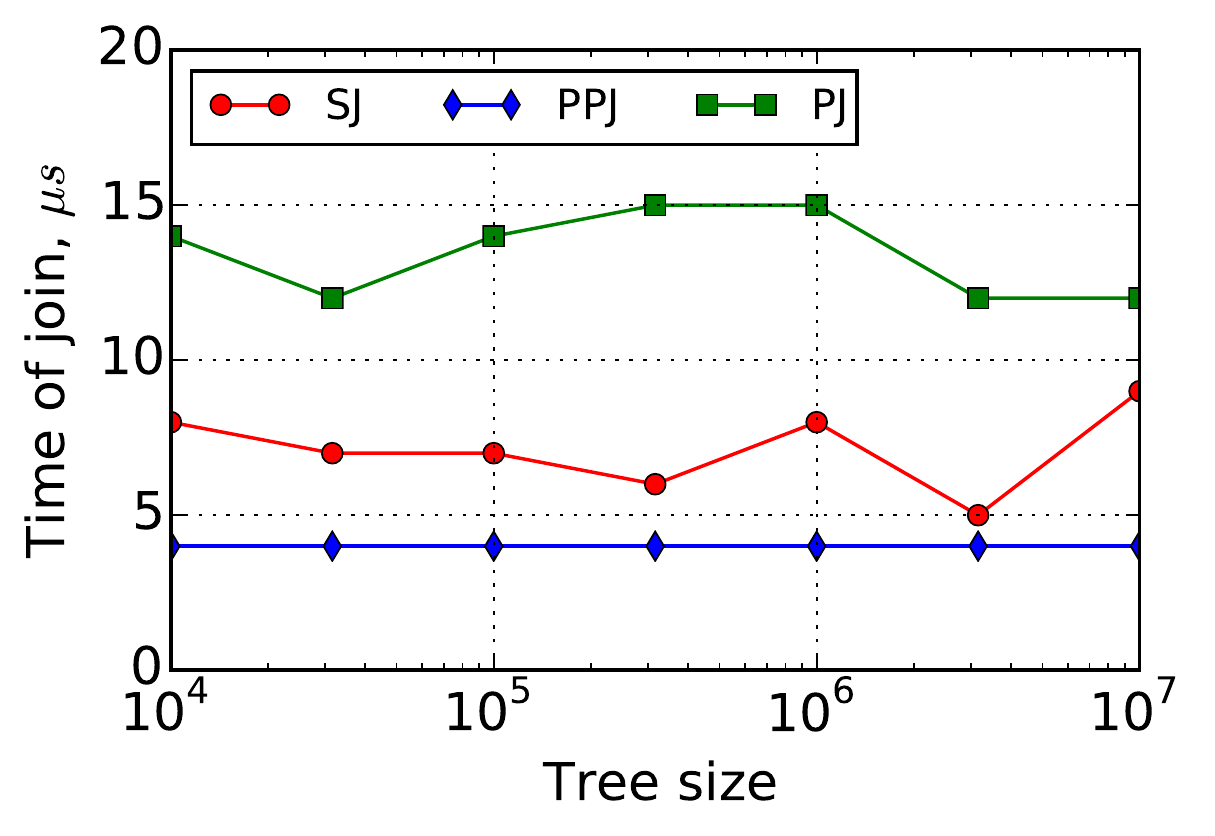}
		\caption{Comparison of join algorithms. $T = \sqrt{10} ^ i, i = 8, \dots, 14$,
			and $p = 31$}
		\label{fig:time_uniform}
	\end{figure}
}{}%

Additionally, we measure the number of visited nodes 
during the course of the join algorithms
to show the theoretical advantage of the the \textbf{PJ} algorithm over the
the \textbf{SJ} and the \textbf{PPJ} algorithms.
Figure~\ref{fig:nodes_number} shows the results of the tests with two initial trees, which are built using
a uniform and a skewed uniform(suffix "\_SU") distribution.

The \textbf{PJ} algorithm visits significantly less nodes 
than the \textbf{SJ} algorithm (by a factor of
$2.8$) on the tests with a uniform distributions. But it visits almost
the same number of nodes as the \textbf{PPJ} algorithm.
The \textbf{PJ\_SU} algorithm visits less nodes than the \textbf{SJ\_SU} and the \textbf{PPJ\_SU} 
algorithms 
(by factor of $3.9$ and $1.3$, respectively). We explain this by the fact that 
a skewed uniform distribution constructs an initial tree that is split into $31$ trees,
such that the first tree is significantly higher than other trees. Therefore,
the algorithms \textbf{SJ} and \textbf{PPJ} visit more nodes than the algorithm \textbf{PJ}
during the access to a spine node by rank.
This suggests that the \textbf{PJ} algorithm may outperform its competitors on instances with 
deeper
trees and/or additional work per node (for example, an I/O operation per node in external B-tree).

\iftoggle{long}{%
	\begin{figure}
		\centering
		\includegraphics[scale = 0.60]{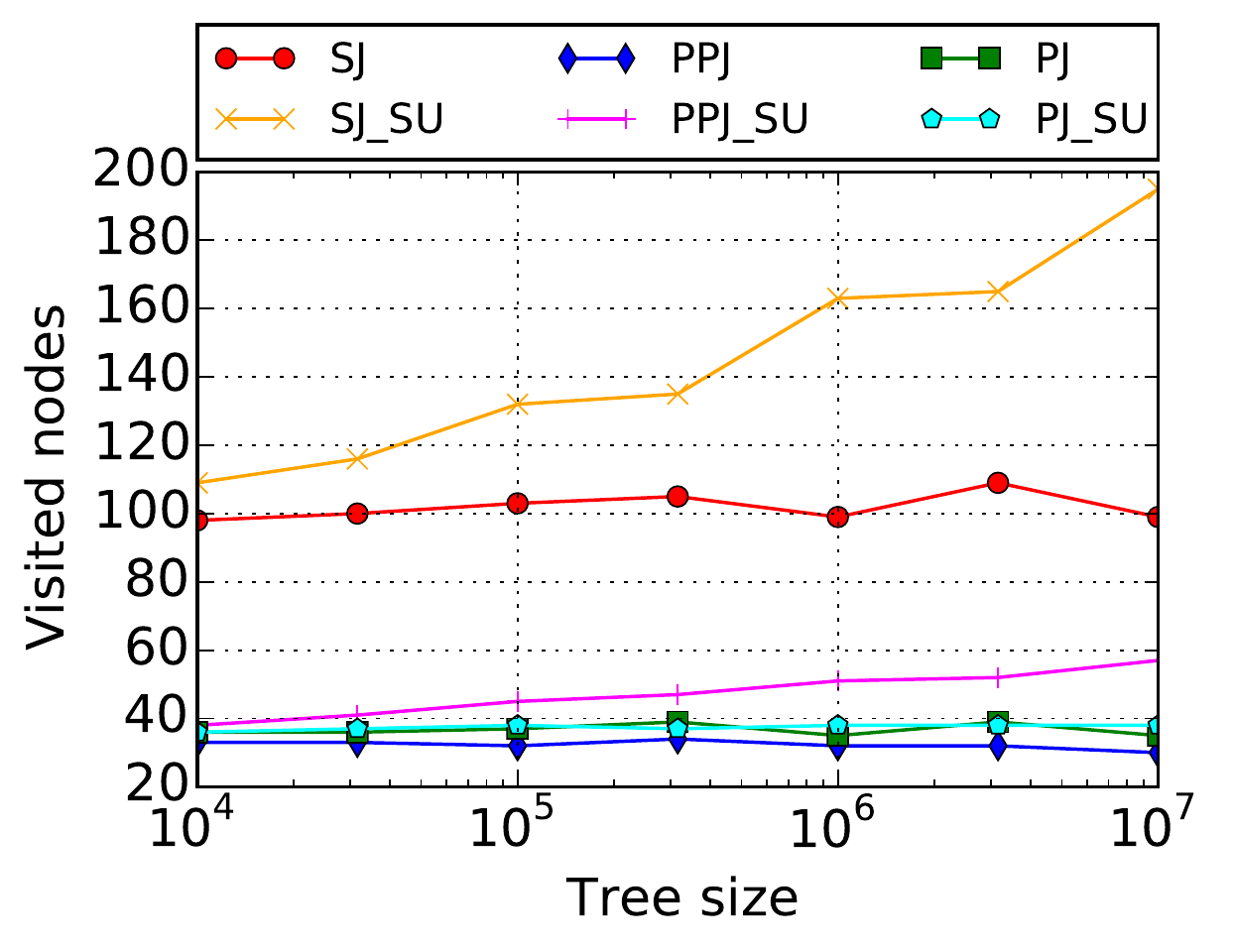}
		\caption{Visited nodes of the join algorithms.
			$T = \sqrt{10} ^ i, i = 8, \dots, 14$, and $p = 31$}
		\label{fig:nodes_number}
	\end{figure}
}{}%

\iftoggle{long}{}{%
	\begin{figure*}[!htb]
		\minipage{0.35\textwidth}
		\includegraphics[width=\linewidth]{Split_Btree_4_8_sz_10K_100M_q_100_it_1000_p_31.pdf}
		\caption{Split algorithms. \\$T = \sqrt{10} ^ i, i = 8, \dots, 14$, and $p = 31$}
		\label{fig:split_test}
		\endminipage
		\minipage{0.35\textwidth}
		\includegraphics[width=\linewidth]{Join_Btree_4_8_sz_10K_100M_q_100_it_1000_p_31.pdf}
		\caption{Join algorithms. \\$T = \sqrt{10} ^ i, i = 8, \dots, 14$, $p = 31$}
		\label{fig:time_uniform}
		\endminipage
		\minipage{0.35\textwidth}%
		\includegraphics[width=\linewidth]{Join_Visited_Nodes_Btree_4_8_sz_10K_100M_q_100_it_1000_p_31.pdf}
		\caption{Visited nodes of the join algorithms. $T = \sqrt{10} ^ i, i = 8, \dots, 14$, $p = 31$}
		\label{fig:nodes_number}
		\endminipage
		\hfill
		
		\minipage{0.35\textwidth}
		\includegraphics[width=\linewidth]{Btree_4_8_sz_100M_q_16_10M_it_10000_300_p_16.pdf}
		\caption{Bulk updates algorithms.}
		\label{fig:time_per_elem}
		\endminipage
		\centering
		\minipage{0.35\textwidth}
		\includegraphics[width=\linewidth]{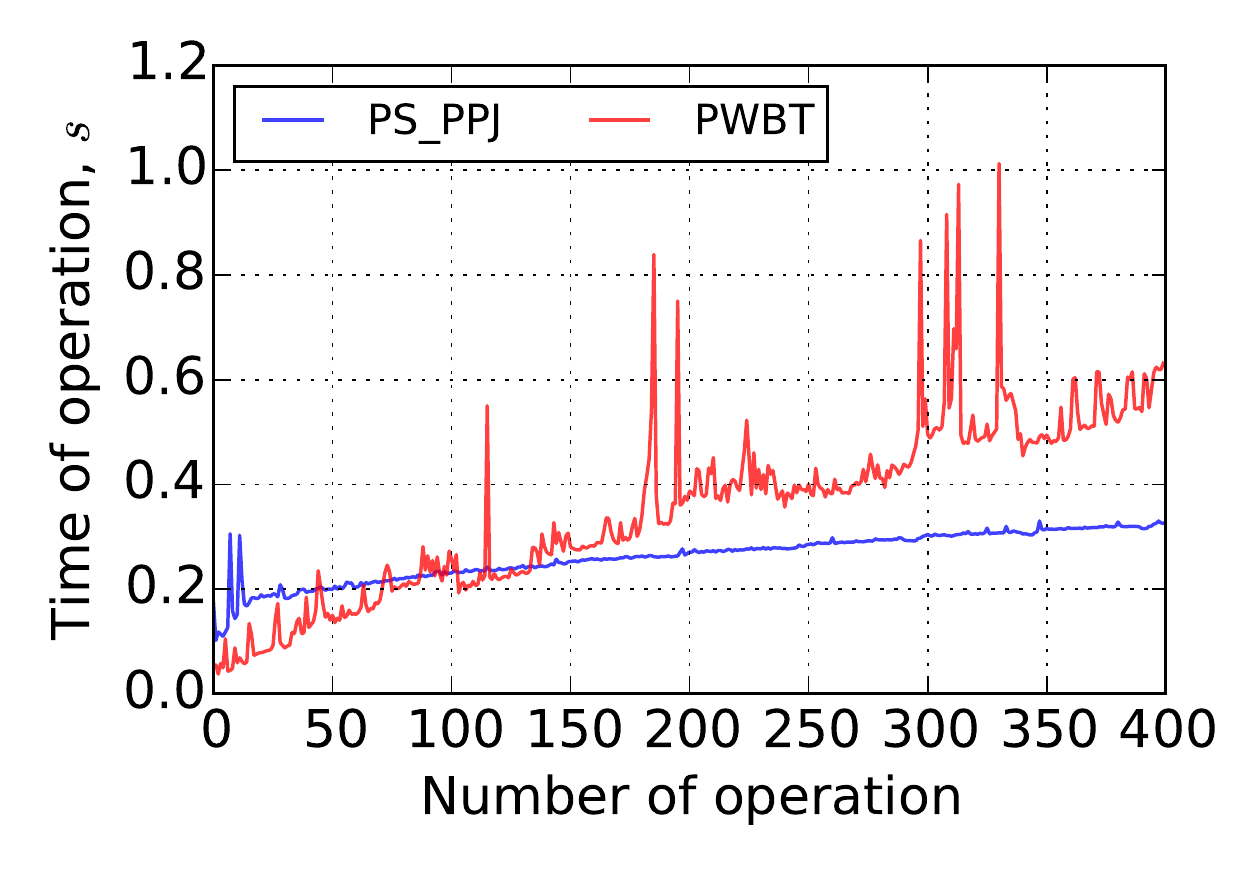}
		\caption{Running time of all operations. $T = 100M, B = 10M$ and $I = 400$}
		\label{fig:time_ops}
		\endminipage
	\end{figure*}
}%

\paragraph{Comparison of parallel search trees}
Here we compare our sequential and parallel implementations of search trees
and four competitors:%
\iftoggle{long}{%
	\begin{enumerate}
		\item \textbf{Seq} is a $(4, 8)$-tree with a sequential bulk updates from
		Section~\ref{paragraph:seq_merge};
		\item \textbf{PS\_PPJ} is a $(4, 8)$-tree with a \textbf{parallel split} phase, 
		a union phase and a \textbf{parallel pairwise join} 
		phase;
		\item \textbf{PWBT} is a Parallel Weight-Balanced B-tree~\cite{EKS14};
		\item \textbf{PRBT} is a Parallel Red-Black Tree~\cite{FriasS07};
		\item \textbf{PBST} is a Parallel Binary Search Tree~\cite{DBLP:journals/corr/BlellochFS16} (we compile it using
		\texttt{g++-4.8} compiler with \texttt{Cilk} support).
	\end{enumerate}
}{%
~(1) \textbf{Seq} is a $(4, 8)$-tree with a sequential bulk updates from
Section~\ref{paragraph:seq_merge};
(2) \textbf{PS\_PPJ} is a $(4, 8)$-tree with a \textbf{parallel split} phase, 
a union phase and a \textbf{parallel pairwise join} 
phase;
(3) \textbf{PWBT} is a Parallel Weight-Balanced B-tree~\cite{EKS14};
(4) \textbf{PRBT} is a Parallel Red-Black Tree~\cite{FriasS07};
(5) \textbf{PBST} is a Parallel Binary Search Tree~\cite{DBLP:journals/corr/BlellochFS16} (we 
compile it using
\texttt{g++-4.8} with \texttt{Cilk} support).
}

Figure~\ref{fig:time_per_elem} shows the measurements of the tests with a uniform distribution, 
where
$T = 100M, B = 16, \sqrt{10}^i, i = 4, \dots, 14$
($I = 10000$ for $B = 16, \dots, 100K$ and $I = 4G / B$ for $B = 316227, \dots, 10M$)
and $p = 16$ (we use 16 cores, since the \textbf{insert} phase is cache-efficient and
hyper-threading does not improve performance). We achieve relative speedup up to 12 over
sequential $(a,b)$-trees. Even for very small batches of size 100 we
observe some speedup. Note that the speed-ups of the \textbf{join} and
\textbf{split} phases are less than 12. But they do not affect the
total speed-up, since the time of the insertion phase dominates them.

On average, our algorithms are $1.8$ times faster than a \textbf{PWBT}. We
outperform the \textbf{PWBT} since a subtree rebalancing (linear of the subtree size)
can occur during an insert operation.
Also our data structure is $5.5$ times faster than a \textbf{PRBT}, $5.7$
times faster than a \textbf{PBST} and $10.7$ faster than a \textbf{Seq}. Note that
the \textbf{PRBT} and the \textbf{PBST} failed in the last four tests due to the lack of memory 
space.
Hence, we expect even greater speed-up of our algorithms compare to them.
For small batches, the speedup is larger which we attribute to the startup overhead
of using Cilk.

Additionally, we run the tests for the \textbf{PS\_PPJ} and the fastest competitor
the \textbf{PWBT} with the same parameters using a normal,
skewed uniform and increasing uniform distributions to generate bulk updates.
They perform faster in these tests than in the tests with
a uniform distribution due to the improved cache locality.
Although a subtree rebalancings in \textbf{PWBT} last longer than in the tests
with a uniform distributions, they occur less frequently.
On average, the \textbf{PS\_PPJ} is
$2.0$, $2.6$ and $2.1$ times faster than the \textbf{PWBT} in tests with a normal,
skewed uniform and increasing uniform distributions respectively.

Figure~\ref{fig:time_ops} shows another comparison of the \textbf{PS\_PPJ}
with the \textbf{PWBT}.
This plot shows the worst-case guaranties of the \textbf{PS\_PPJ}.
The spikes on the plot of the \textbf{PWBT} are due to the amortized cost
of operations.
We conclude that the \textbf{PS\_PPJ} is preferable in real-time applications
where latency is crucial.


\iftoggle{long}{%
	\begin{figure}
		\centering
		\includegraphics[scale = 0.60]{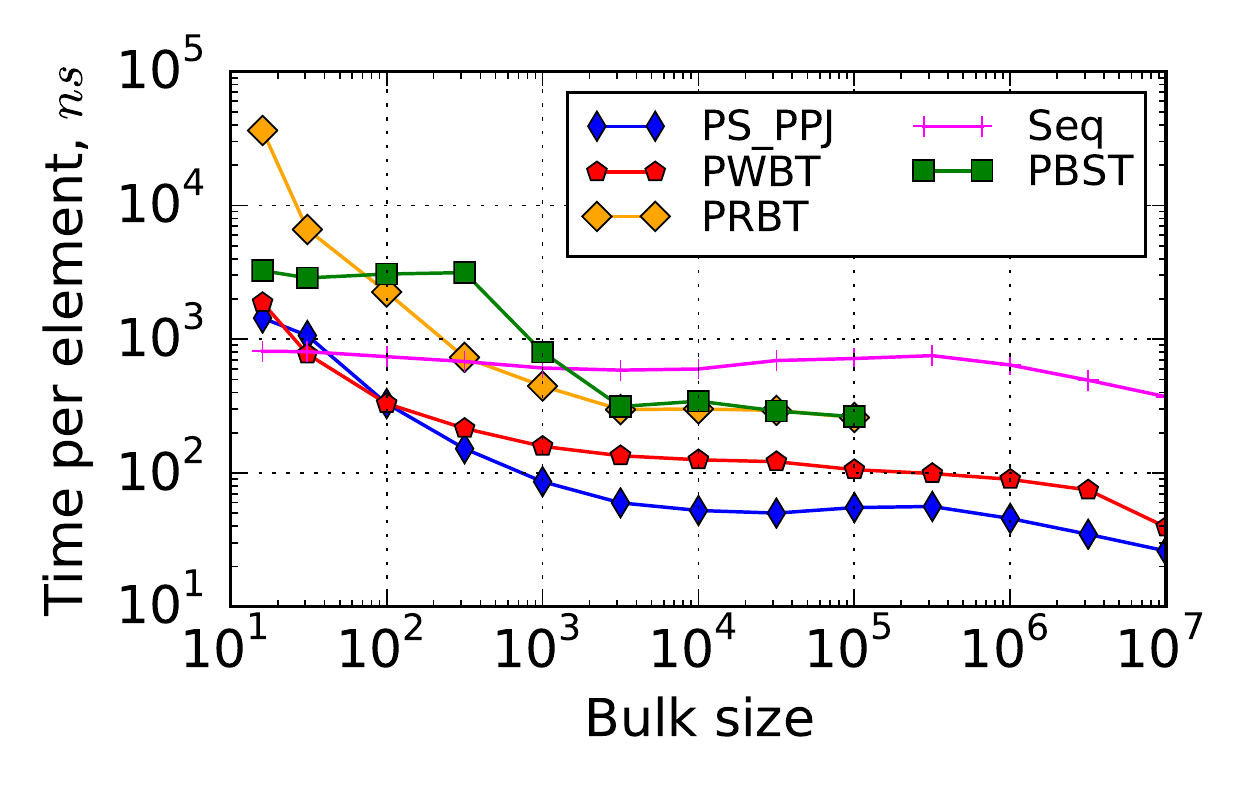}
		\caption{Bulk updates algorithms.}
		\label{fig:time_per_elem}
	\end{figure}
}{}%

\iftoggle{long}{%
	\begin{figure}
		\centering
		\includegraphics[scale = 0.60]{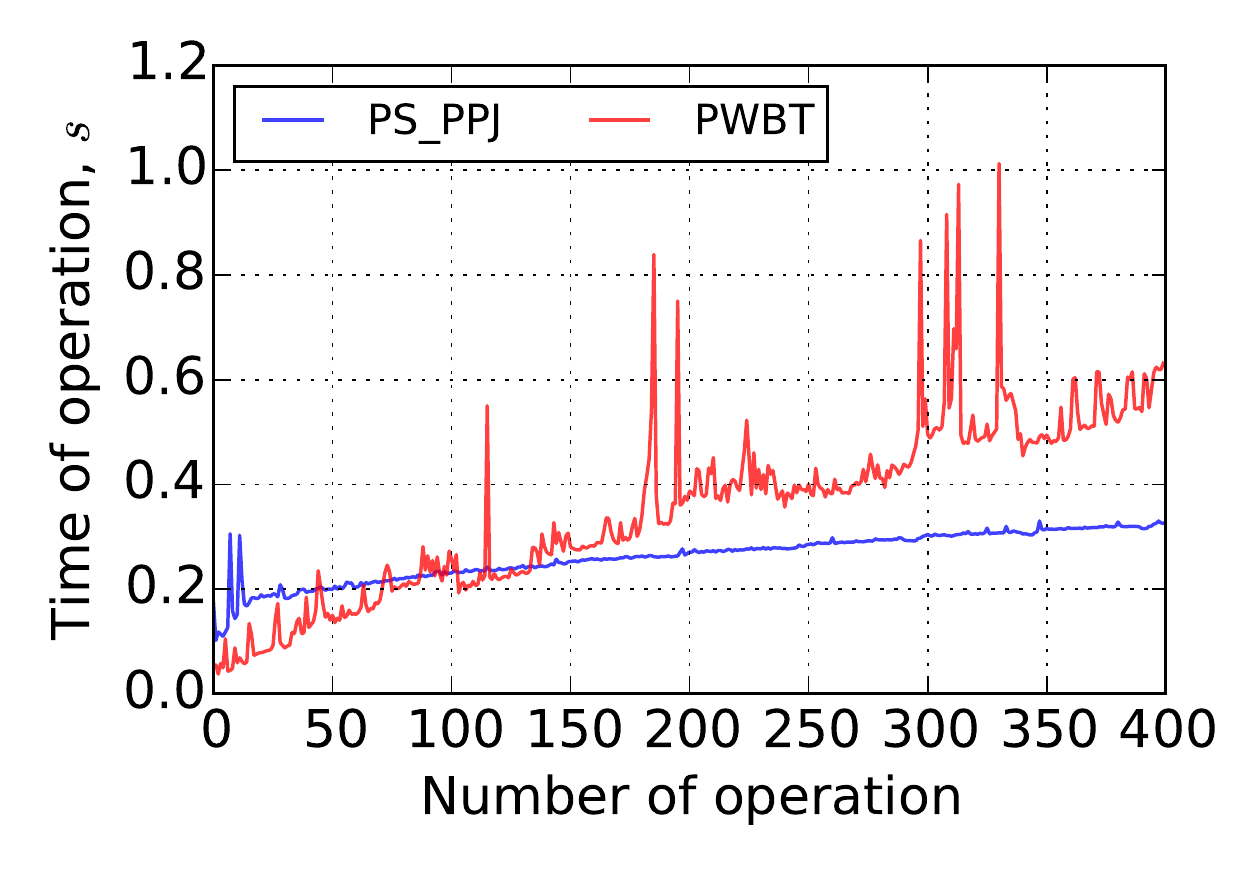}
		\caption{Running time of all operations. $T = 100M, B = 10M$ and $I = 400$}
		\label{fig:time_ops}
	\end{figure}
}{}%


\bibliographystyle{amsplain}
\bibliography{soda,diss}
\end{document}